%% file: main.tex
\documentclass[acmsmall, screen, nonacm]{acmart}

\newif\ifarxiv
\arxivtrue

\input{style}

\begin{document}

\title{Synthesizing Backward Error Bounds, Backward}
\ifarxiv\else
    \titlenote{Proofs and additional details may be found in the supplemental materials \cite{supplement}.}
\fi

\author{Laura Zielinski}
\orcid{0009-0006-6516-1392}
\affiliation{
  \institution{Cornell University}
  \city{Ithaca}
  \country{USA}
}
\email{lzielinski@cs.cornell.edu}

\author{Justin Hsu}
\orcid{0000-0002-8953-7060}
\affiliation{
  \institution{Cornell University}
  \city{Ithaca}
  \country{USA}
}
\email{justin@cs.cornell.edu}

\begin{abstract}
    Backward stability is a desirable property for a well-designed numerical algorithm: given an input, a backward stable floating-point program produces the exact output for a nearby input. While automated tools for bounding the forward error of a numerical program are well-established, few existing tools target backward error analysis. We present a formal framework that enables sound, automated backward error analysis for a broad class of numerical programs. First, we propose a novel generalization of the definition of backward stability that is both compositional and flexible, satisfied by a wide range of floating-point operations. Second, based on this generalization, we develop the category $\Shel$ where morphisms model stable numerical programs, and show that structures in $\Shel$ support a rich variety of backward error analyses. Third, we implement a tool, $\eggshel$, that automatically searches within a syntactic subcategory of $\Shel$ to prove backward stability for a given program. Our algorithm handles many programs with variable reuse, a known challenge in backward error analysis. We prove soundness of our algorithm and use our tool to synthesize backward error bounds for a suite of programs that were previously beyond the reach of automated analysis.
\end{abstract}




\maketitle

\section{Introduction}
\input{sections/intro}

\section{Background}
\input{sections/background}

\section{Relational Backward Error Lenses}\label{sec:lenses}
\input{sections/lenses}

\section{Shel: The Category of Shifted Error Lenses}\label{sec:category}
\input{sections/category}

\section{Case Studies}\label{sec:case_studies}
\input{sections/case_studies}

\section{eggshel Implementation}\label{sec:implementation}
\input{sections/implementation}

\section{Related Work}
\input{sections/related}

\section{Conclusion and Future Directions}
\input{sections/conclusion}

\begin{acks}
  We thank the anonymous reviewers for their helpful comments. This project was inspired by discussions at the Adjoint School 2025 in-person research week. We are grateful to the Adjoint School organizers for the opportunity, and we thank Jonas Forster, Chase Ford, Max Fan, J\'{o}n H\'{a}kon Gar{\dh}arsson, Jessica Richards, Aven Dauz, and the other members of the research group for productive discussions. This work was partially supported by a grant from the Office of Naval Research (N00014-23-1-2134).
\end{acks}

\section*{Artifact}
The implementation $\eggshel$ described in \zcref{sec:implementation} is publicly available \cite{eggshel}.

\makeatletter
\let\@vspace\@vspace@orig
\makeatother

\bibliographystyle{ACM-Reference-Format}
\bibliography{header, references}

\newpage\appendix\input{sections/appendix}

\end{document}

%% file: style.tex
\usepackage{quiver}
\usepackage{listings}
\usepackage{mathpartir}
\usepackage{multirow}
\usepackage{booktabs}
\usepackage{amsthm}
\usepackage{thm-restate}
\usepackage{zref-clever}
\zcsetup{cap}

\lstdefinelanguage{egglog}{
    basicstyle=\ttfamily\small,
    keywordstyle=\textbf,
    keywords={datatype, Expr, String, f64, Ctx, rewrite, relation},
    morecomment=[l]{;},
    commentstyle=\color{gray},
    showstringspaces=false,
    xleftmargin=2em
}
\newcommand{\Bean}{\textsc{Bean}}
\newcommand{\eggshel}{\textbf{eggshel}}

\newcommand{\Shel}{\textbf{Shel}}
\newcommand{\egglog}{\texttt{egglog}}

\newcommand{\tf}{\tilde{f}}
\newcommand{\tg}{\tilde{g}}
\newcommand{\tx}{\tilde{x}}
\newcommand{\ty}{\tilde{y}}

\newcommand{\R}{\mathbb{R}}
\newcommand{\F}{\mathbb{F}}
\newcommand{\N}{\mathbb{N}}
\newcommand{\e}{\varepsilon}
\newcommand{\RR}{\mathcal{R}}

\newcommand{\op}{\mathbf{op}}
\newcommand{\rnd}{\mathbf{rnd}}
\newcommand{\sqr}{\mathbf{sqrt}}

\newcommand{\add}{\mathbf{add}}
\newcommand{\sub}{\mathbf{sub}}
\newcommand{\mul}{\mathbf{mul}}
\newcommand{\dmul}{\mathbf{dmul}}
\newcommand{\divi}{\mathbf{div}}

\newcommand{\share}{\mathbf{share}}
\newcommand{\push}{\mathbf{push}}

\newcommand{\denot}[1]{\llbracket{#1}\rrbracket}
\newcommand{\fixedarrow}[1]{\xrightarrow{\makebox[50pt][c]{$\scriptstyle #1$}}}

\DeclareMathOperator{\sign}{sign}
\DeclareMathOperator{\id}{id}
\DeclareMathOperator{\Var}{Var}
\DeclareMathOperator{\FV}{FV}
\DeclareMathOperator{\dom}{dom}

\theoremstyle{plain}
\newtheorem{theorem}{Theorem}[section]
\newtheorem{definition}[theorem]{Definition}
\newtheorem{lemma}[theorem]{Lemma}
\newtheorem{remark}[theorem]{Remark}

\theoremstyle{definition}
\newtheorem{example}[theorem]{Example}

%% file: sections/intro.tex
Numerical computation requires approximating real numbers with floating-point values, inevitably introducing rounding error at each operation. One of the standard ways to check if a numerical algorithm is well-designed is by proving \emph{backward stability}, which states that a program's output is the mathematically exact output corresponding to a nearby input; the distance between the actual and the nearby input is called the \emph{backward error}. Backward error analysis provides a way to measure the reliability of an algorithm while taking into account the natural tendency of the problem to amplify errors \cite[p.~9]{higham}. It is a dual notion to forward error, which is the difference between the computed and the ideal output. 

While there are now many automated tools for bounding forward error (e.g., \cite{fptaylor, precisa, gappa}), few have targeted backward error. A central challenge is that backward stability of subprograms does not guarantee backward stability of a full program. In particular, reusing variables across multiple computations often breaks backward stability. Classical analyses tend to use ad hoc reasoning \cite{higham} which is difficult to automate. A notable exception is $\Bean$ \cite{bean}, a programming language whose type system ensures that any well-typed program is backward stable. $\Bean$ demonstrated that a static and sound backward error analysis is feasible, but its affine type system severely limits the programs it can analyze and it cannot type many backward stable programs that reuse variables.

In this work, we present a novel automated backward error analysis that can handle a significantly broader class of programs and proofs of backward stability. Central to our analysis, we introduce a new, stronger version of backward stability that behaves well under composition while capturing a broad variety of stable floating-point operations. We frame our definition within a novel category $\Shel$, where morphisms are backward stable programs, and demonstrate how backward stable programs that reuse variables are naturally accommodated in $\Shel$. Guided by this category, we develop a tool called $\eggshel$ that searches within $\Shel$ to prove a given program's backward stability, automatically deriving backward error bounds. 

\paragraph{Contribution 1: Relational Backward Error Lenses} In \zcref{sec:lenses}, we present a stronger definition of backward stability called a \emph{relational} backward error lens, which serves as the foundation of our analysis. Like the original backward error lenses proposed by \citet{bean}, our definition is designed to compose well and carry a proof of backward stability. However, our definition is more general and captures new stable operations, such as correctly-rounded square root and logarithm.

\paragraph{Contribution 2: Categorical Semantics} In \zcref{sec:category}, we develop a category $\Shel$ where morphisms are relational backward error lenses, and objects give us fine-grained information about how inputs are perturbed in the course of backward error analysis. For instance, we introduce the \emph{push product} to model variables with correlated backward error, which is crucial for capturing backward stable programs that reuse variables. In \zcref{sec:case_studies}, we use these novel categorical structures to form lenses corresponding to concrete examples of backward stable programs. We also show how $\Shel$ can support multiple analyses of the same program, where different backward error bounds can be established through different choices of lenses.

\paragraph{Contribution 3: Implementation} In \zcref{sec:implementation}, we develop an automated analysis called $\eggshel$ that can search for backward error lenses proving backward stability for a given target program. Our tool is implemented in $\egglog$~\cite{egglog} and relies on a concrete, syntactic presentation of a subcategory of $\Shel$. We prove soundness of our tool and show that it can automatically derive backward error bounds for programs that are not supported by existing tools for backward error analysis.

%% file: sections/background.tex
In this section, we introduce our model of rounding error, give the standard definition of backward stability, and present an example of a backward stability analysis.

\subsection{Floating-Point Error}
The floating-point numbers $\F$ form a finite subset of the real numbers exactly representable on a computer. Compared to ideal mathematical functions, operations on floating-point numbers inherently introduce rounding error, and an \emph{error model} is a mathematical specification describing the magnitude of these errors. We use the model proposed by \citet{olver}, which is parametrized by the \emph{unit roundoff} $u$, a constant that depends on the rounding mode and floating-point precision. For example, under round-to-nearest with double precision (binary64), $u=2^{-53}$. We assume a fixed rounding mode and precision and give our bounds in terms of the constant $\e=u/(1-u)$.

\begin{definition}[Olver's Rounding Error Model]\label{def:olver_model}
    For floating-point numbers $x,y\in \F$ and a primitive floating-point operation $\op$ modeling an exact $op\in\{+,-,\times,\div\}$, assuming no underflow or overflow, 
    \begin{equation*}
        \op~x~y =(x~op~y)e^\delta\quad\text{with}\quad |\delta|\leq \e =\frac{u}{1-u},
    \end{equation*}
    where $u$ is the unit roundoff. We also assume $\sqr~x=\sqrt{x}e^\delta$ where $|\delta|\leq \e$.
\end{definition}

In other words, the primitive floating-point operations $\add$, $\sub$, $\mul$, $\divi$, and $\sqr$ are \emph{correctly-rounded}. This error model is satisfied by machines complying with the IEEE-754 standard when using a supported rounding mode such as round-to-nearest \cite{mech_olver}; we use Olver's model as it simplifies the calculation of error bounds \cite{olver}. Intuitively, for small $\delta$ the factor $e^\delta$ is very close to $(1 + \delta)$, so a primitive floating-point operation returns the exact result perturbed by a small factor. 

Olver's rounding error model can be formulated in terms of an extended metric on $\R$ \cite{olver}:

\begin{definition}[Relative Precision Metric]\label{def:rp_metric}
    For $x,y\in\R$, we define $RP(x,y)$ as
    \begin{equation*}
        RP(x,y) = \begin{cases}
            0 & x=y=0 \\
            |\ln(x/y)| & \sign(x)=\sign(y) \\
            \infty & \text{else.}
        \end{cases}
    \end{equation*}
\end{definition}

For example, the $RP$ distance between nonzero $x\in\R$ and $\rnd~x=xe^\delta$ is 
\begin{align*}
    RP(x,xe^\delta)&=\left|\ln\frac{x}{xe^\delta}\right| = |\ln(x)-(\ln(x) + \ln(e^\delta))| = |\delta|\leq \e.
\end{align*}
This metric is useful as it offers a notion of relative distance; note that the usual definition of relative error, $|x-y|/|x|$, does not form a true metric, though one can bound the relative error using the $RP$ distance \cite{olver}. We will use the $RP$ metric to measure distance in $\R$ throughout this paper. 

\begin{remark}
    As a convenience, for $x,y\in\R^n$ and a componentwise bound $c\in\R^n_{\geq 0}$, we write $d(x,y)\leq c$ to mean $RP(x_i,y_i)\leq c_i$ for all $i\in\{1,\dots,n\}$.
\end{remark}

\subsection{Backward Stability and Backward Error}
With our rounding error model fixed, we turn to a definition of backward stability \cite[p.~122]{higham}. In this section and the next, $f$ denotes a real function $X\subseteq\R^n\to \R^m$. Floating-point implementations of $f$ are denoted $\tf:X\cap\F^n\to \F^m$. 

\begin{definition}[Componentwise Backward Stability, Standard Definition]\label{def:std_def}
    We say $\tf$ is $\alpha$-\emph{backward stable} for some $\alpha\in\R_{\geq 0}^n$ if, for any floating-point input $x$, there exists a \emph{witness} $\tx\in\R^n$ such that 
    \begin{equation*}
        f(\tx)=\tf(x) \quad \text{and} \quad d(x,\tx)\leq \alpha\e.
    \end{equation*}
    For each $i\in\{1,\dots,n\}$, the constant $\alpha_i\e$ is the \emph{backward error bound} of $\tf$ with respect to input $x_i$.
\end{definition}

If we think of the input $x$ as describing parameters of some problem to be solved, backward stability states that a program $\tf(x)$ computes the exact solution to nearly the right problem: $f(\tx)$. Backward error stands in contrast to \emph{forward} error, which is simply the distance between the computed and ideal outputs. While its definition may appear somewhat unnatural, backward error is often preferred by numerical analysts for several reasons \cite[pp.~ix--x]{corless}. First, forward error depends on the \emph{sensitivity} of the function being computed, which bounds how changes to the input affect the output. Highly sensitive functions---such as subtraction---may have no implementation with small forward error, since even small rounding errors are greatly amplified. Thus, a large forward error may be because the \emph{problem} is inherently difficult to solve accurately, or because the \emph{program} is introducing excessive rounding errors \cite[p.~9]{higham}. 

Backward stability provides a way to analyze the quality of a program, independent of the sensitivity of the underlying problem. A program with a small backward error is considered to be as accurate as can be expected, even if its forward error is large. For example, floating-point subtraction is backward stable, even though it is not forward stable. Furthermore, backward error can be usefully compared to other natural sources of error, such as the measurement uncertainty of the input data \cite[pp.~6--7]{higham}.

The primitive floating-point operations $\add,\sub,\mul,\divi,\sqr$ all immediately satisfy backward stability by \zcref{def:olver_model}. Take $\mul$ for example:
\begin{align*}
    \tf(x,y) &= \mul~x~y = xye^\delta = (xe^{\delta/2})(ye^{\delta/2}) = f(\tx,\ty)
\end{align*}
where $|\delta|\leq\e$ and we set $\tx=xe^{\delta/2}$ and $\ty=ye^{\delta/2}$. Since $RP(x,\tx)=RP(y,\ty)\leq \e/2$, we conclude that $\mul$ satisfies \zcref{def:std_def} and is $(1/2,1/2)$-backward stable.

This argument is a simple example of a \emph{backward error analysis}.  Intuitively, we think of performing a floating-point operation as ``pushing'' some amount of error onto its operands, or ``perturbing'' them somehow. For example, $\mul$ in this analysis pushes $\delta/2$ error onto each operand. If the result of $\mul~x~y$ is used in another operation and has error pushed onto it, that error gets distributed down to $x$ and $y$, perturbing them further. For a worked example of a larger analysis, see \ifarxiv\zcref{app:analysis_by_hand}.\else\cite{supplement}.\fi

For $\mul$ and $\divi$ in particular, there are multiple possible backward error bounds. For example, it is equally valid to say that
\begin{align}\label{eq:dmul}
    \mul~x~y = (xe^{\delta})y = x(ye^{\delta}).
\end{align}
That is, we can distribute the backward error any way we wish, including setting $\tx=x$ (pushing all the error onto $y$ and leaving $x$ unperturbed) or $\ty=y$ (pushing all the error onto $x$ and leaving $y$ unperturbed). In fact, both $\mul$ and $\divi$ are $(p,1-p)$-backward stable for any $p\in[0,1]$. This flexibility significantly complicates backward error analysis---while the choice of how to distribute error is usually not important, for some programs the backward stability argument requires a specific allocation of error.

\subsection{A Motivating Example}
Proving backward stability by hand grows more difficult for larger programs, and it is natural to look for a more compositional method. For example, the backward error analysis for back substitution (solving $Uy=b$ for upper triangular matrix $U$) relies on the analysis of dot product. Compositional analysis offers a structured, algorithmic way to determine backward error bounds, and this is the guiding principle of our analysis.

\citet{bean} designed the first tool exemplifying this approach: a programming language, $\Bean$, whose type system only allows stable programs. However, $\Bean$ is overly conservative and rejects many valid programs. We consider a motivating example of a backward stable program for which the analysis in $\Bean$ fails due to restrictions in the type system and semantics, but for which our analysis succeeds. Our program takes the Euclidean norm of a vector in $\R^2$:
\begin{align*}
    \tf(a,b) &= \sqr~(\add~(\mul~a~a)~(\mul~b~b)).
\end{align*}
While it is not at all obvious, this program is backward stable. However, this program fails to type check in $\Bean$ for two reasons. First, square root, though primitive and backward stable, is not a valid $\Bean$ operation. This is because square root does not satisfy their definition of composable backward stability, which is more restrictive than the standard definition of backward stability (\zcref{def:std_def}). In our work, we find an alternative definition which is satisfied by a wider variety of stable operations, such as square root.

Second, $\Bean$ is an \emph{affine} type system which does not allow variable reuse, except for variables which are never perturbed. Thus, the stable subprograms $\mul~a~a$ and $\mul~b~b$ in our example are disallowed since $a$ and $b$ are both reused and perturbed. To understand the reason for the affine restriction, suppose we have two stable programs $\tf_1$ and $\tf_2$, each with a floating-point input $x$. By backward stability, we know that $\tf_1(x)=f_1(\tx_1)$ and $\tf_2(x)=f_2(\tx_2)$ for some $\tx_1,\tx_2\in\R$.  Now, consider the paired program $(\tf_1,\tf_2)$ that takes a single input $x$ and returns both outputs as a tuple. To satisfy backward stability, we need to find $\tx$ such that 
\begin{equation*}
    (\tf_1(x),\tf_2(x))=(f_1(\tx),f_2(\tx)),
\end{equation*}
but if $\tx_1\neq \tx_2$, there may not exist such an $\tx$. This partially explains why programs like outer product and standard matrix-matrix multiplication are not backward stable---they reuse variables many times, each computation perturbing the inputs in potentially different directions. However, there are backward stable programs that \emph{do} reuse variables and perturb them. If both $\tf_1$ and $\tf_2$ have the same witness $\tx$, then the program $(\tf_1,\tf_2)$ is indeed backward stable. Our analysis captures these programs, which are valid when all subprograms push the \emph{same} error onto a common variable.

%% file: sections/lenses.tex
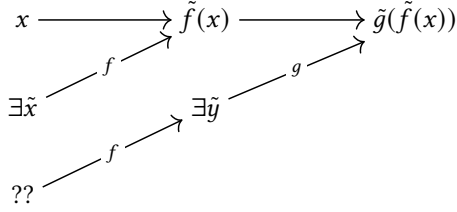
\begin{figure}
  \centering
  \input{figures/compose.tex}
  \caption{The standard definition of backward stability fails to compose.}
  \Description{A diagram showing how the standard definition of backward stability fails to compose.}
  \label{fig:fails_to_compose}
\end{figure}

Our broad goal is to prove backward error guarantees compositionally.  Concretely, starting from the primitive backward stable floating-point operations, we want to leverage the stability of subprograms to prove the stability of a large program. However, the standard definition of backward stability is not preserved under composition. To see why, suppose we are given two backward stable programs $\tf$ and $\tg$, and we would like to show that the program $\tg\circ\tf$ modeling the exact function $g\circ f$ is backward stable. For an input $x$, we require an $\tx$ such that
\begin{equation*}
    g(f(\tx))=\tg(\tf(x)).
\end{equation*}
Applying the backward stability property of $\tg$ at input $\tf(x)$ gives a perturbed input $\ty$, as depicted in \zcref{fig:fails_to_compose}. But now we are stuck---the backward stability guarantee for $\tf$ only provides a witness $\tx$ satisfying $f(\tx)=\tf(x)$, not $f(\tx)=\ty$ for an arbitrary value $\ty$. Pictorially, backward stability of $\tf$ shows that there is an input $\tx$ at the middle-left in \zcref{fig:fails_to_compose}, but establishing backward stability of $\tg \circ \tf$ requires an input at the bottom-left. Thus in order for backward stability to compose, we need a stronger version that guarantees we can find perturbed inputs mapping to values besides the computed output.

\subsection{Backward Error Lenses in \textsc{Bean}}

\citet{bean} offered the first definition of backward stability that composes.  Their definition pairs the exact and approximate computations $f$ and $\tf$ with a \emph{backward map}: a function $b$ which, given an input $x$ and a value $y$ in the output space, returns a witness $\tx$ such that $f(\tx)=y$. The distance between the witness $\tx$ and the original input $x$ must also be bounded by the distance between $\tf(x)$ and $y$. Then, we can complete the diagram in \zcref{fig:fails_to_compose} by taking $b(x,\tilde{y})$ to be our perturbed input, establishing backward stability. \citet{bean} called such triples $(f, \tf, b)$ \emph{backward error lenses}. In order to bound the distance between the actual and perturbed inputs, backward error lenses impose an additional \emph{non-expansiveness} condition; in this work, we call their definition \emph{non-expansive lenses}.

\begin{definition}[Non-Expansive Backward Error Lens]\label{def:bean_def}
    The triple $(f,\tf,b)$, where $b:\F^n\times \R^m\rightharpoonup \R^n$ is a backward map, is an $(\alpha,\beta)$-\emph{non-expansive backward error lens} for $\alpha\in\R_{\geq 0}^n$ and $\beta\in\R_{\geq 0}^m$ if for any floating-point input $x$ and value $y\in\R^m$ at finite distance from $\tf(x)$, we have
    \begin{equation*}
        f(\tx)=y \quad \text{and} \quad \max_{i\in\{1,\dots,n\}}RP(x_i,\tx_i) -\alpha_i\e\leq \max_{j\in\{1,\dots,m\}}RP(\tf(x)_j,y_j)-\beta_j\e
    \end{equation*}
    where $\tx=b(x,y)$.
\end{definition}

Note that $b$ is a partial map as it is only defined for $y$ at finite distance from $\tf(x)$. We briefly comment on a few aspects of the definition. First, if we instantiate $y$ with $\tf(x)$, we can see that a non-expansive lens satisfies the standard definition from \zcref{def:std_def}, so an $(\alpha,\beta)$-non-expansive lens is $\alpha$-backward stable.  Second, non-expansive refers to the inequality: as $y$ moves farther from $\tf(x)$, $\tx$ is allowed to stray from $x$ by at most the same amount.  Finally, the constants $\alpha$ and $\beta$ play a role when composing lenses: given an $(\alpha,\beta)$-non-expansive lens and a $(\beta,\gamma)$-non-expansive lens, their composition is an $(\alpha,\gamma)$-non-expansive lens \cite{bean}.

\begin{example}
    For any $p\geq 0$, we have a $((p+1,p+1),p)$-non-expansive lens for addition:
    \begin{itemize}
        \item $f:\R^2\to\R$ with $f(x_1,x_2)=x_1+x_2$,
        \item $\tf:\F^2\to\F$ with $\tf(x_1,x_2)=(x_1+x_2)e^\delta$ and $|\delta|\leq\e$,
        \item $b:\F^2\times \R\rightharpoonup \R^2$ with 
        \begin{equation*}
            b((x_1,x_2),y)=\begin{cases}
                (x_1,x_2) & \text{if }x_1+x_2=y=0, \\[5pt]
                \left(\frac{x_1y}{x_1+x_2},\frac{x_2y}{x_1+x_2}\right) & \text{otherwise.}
            \end{cases}
        \end{equation*}
    \end{itemize}
    When $x_1+x_2$ is nonzero, we see that 
    \begin{equation*}
        f(b((x_1,x_2),y)) = \frac{x_1y}{x_1+x_2}+\frac{x_2y}{x_1+x_2}=y,
    \end{equation*}
    so given $y$, the backward map $b$ returns compatible witnesses $\tx_1$ and $\tx_2$. One can verify that the non-expansiveness condition holds with the stated constants.
\end{example}

The four basic operations ($\add$, $\sub$, $\mul$, $\divi$) satisfy \zcref{def:bean_def}, but the non-expansive condition rules out many backward stable programs. For instance, the primitive operation $\sqr$ is not a non-expansive lens---it is possible to show that the distance between the witness $\tx$ and $x$ must be larger than the distance between $\tf(x)$ and $y$. 

\subsection{Relational Backward Error Lenses}
In this work, we present a new definition of composable backward error lenses.  Our definition is weaker than that of \citet{bean}, while still satisfying composition and implying backward stability. Looking at \zcref{fig:fails_to_compose} again, we see that we do not need to find $\tx$ for \emph{every} possible $y$; we only need to handle $y$ \emph{near} $\tf(x)$.  Moreover, the distance between $\tx$ and $x$ need not be bounded by the distance between $\tf(x)$ and $y$. We call our definition \emph{relational lenses}, as we no longer require the non-expansive condition. 
\ifarxiv
Omitted proofs from this section may be found in \zcref{app:lenses_proofs}.
\fi

\begin{definition}[Relational Backward Error Lens]\label{def:shel_def}
    The triple $(f,\tf,b)$ is an $(\alpha,\beta)$-relational backward error lens if for any floating-point input $x$ and $y\in\R^m$ where $d(\tf(x),y)\leq \beta\e$, we have
    \begin{equation*}
        f(\tx)=y \quad \text{and} \quad d(x,\tx)\leq \alpha\e
    \end{equation*}
    where $\tx=b(x,y)$.
\end{definition}
As with non-expansive lenses, by instantiating $y$ with $\tf(x)$, we recover the standard definition. 

\begin{theorem}\label{thm:lens_sound}
    An $(\alpha,\beta)$-relational lens is $\alpha$-backward stable.
\end{theorem}
\begin{proof}
    Given a relational lens, we prove the conditions of \zcref{def:std_def}. For input $x$, let $\tx=b(x,\tf(x))$. By the relational definition, since $d(\tf(x),\tf(x))\leq \beta\e$, we have
    \begin{equation*}
        f(\tx)=\tf(x) \quad\text{and}\quad d(x,\tx)\leq \alpha\e. \qedhere
    \end{equation*} 
\end{proof}

Relational lenses compose and therefore can be used as the framework of our static error analysis.

\begin{restatable}{theorem}{lensComp}\label{thm:lens_comp}
    Given $(\alpha,\beta)$- and $(\beta,\gamma)$-relational lenses with compatible domains, their composition is an $(\alpha,\gamma)$-relational lens.
\end{restatable}

Intuitively, it is easier to show the relational property than the non-expansive property since we can assume a bound on the distance from $y$ to $\tf(x)$. Given this bound $\beta$, we can adjust the bound $\alpha$ as needed; perhaps $\tx$ and $x$ are at most twice as far apart as $y$ and $\tf(x)$, so $\alpha$ is twice $\beta$. Our new definition is strictly weaker than that of \citet{bean} and thus allows more programs.

\begin{theorem}
    An $(\alpha,\beta)$-non-expansive lens $(f,\tf,b)$ is also an $(\alpha,\beta)$-relational lens.
\end{theorem}
\begin{proof}
    Given a non-expansive lens, we prove the conditions of a relational lens. Take an input $x$ and value $y$ such that $d(\tf(x),y)\leq \beta \e$. Since $y$ and $\tf(x)$ are at finite distance, by the non-expansive definition, $f(\tx)=y$. Moreover, for each $k\in\{1,\dots,n\}$, 
    \begin{equation*}
        RP(x_k,\tx_k) -\alpha_k\e\leq\max_{i\in\{1,\dots,n\}}RP(x_i,\tx_i)-\alpha_i\e\leq \max_{j\in\{1,\dots,m\}}RP(\tf(x)_j,y_j)-\beta_j\e\leq 0,
    \end{equation*}
    the final inequality following from the condition that $d(\tf(x),y)\leq \beta\e$. By transitivity, $RP(x_k,\tx_k)\leq \alpha_k\e$, and this is true for each $k$, so we conclude $d(x,\tx)\leq \alpha\e$, satisfying the relational definition.
\end{proof} 

However, the converse is not true: there are relational lenses that are not non-expansive lenses. 

\subsection{Examples of Relational Lenses}\label{sec:lens_ex}
To show that our new generality is useful, we define relational lenses for square root and natural logarithm operations, for which we cannot build non-expansive lenses.

\begin{restatable}{example}{sqrtLens}\label{ex:sqrt_lens}
    For any $p\geq 0$, we have a $(2p+2,p)$-relational lens for square root:
    \begin{itemize}
        \item $f:\R\to\R$ with $f(x)=\sqrt{|x|}$,
        \item $\tf:\F\to\F$ with $\tf(x)=\sqrt{|x|}e^\delta$ where $|\delta|\leq \e$ (from our model in \zcref{def:olver_model}),
        \item $b:\F\times\R\rightharpoonup\R$ with $b(x,y)=y^2$ if $x\geq 0$, otherwise $-y^2$.
    \end{itemize}
\end{restatable}

Next, we assume that we have a correctly rounded natural logarithm function; we let $a$ be the upper bound on the floating-point numbers $\F$.
\begin{restatable}{example}{logLens}
    For any $p\geq 0$, we have a $(3a(p+1),p)$-relational lens for logarithm:
    \begin{itemize}
        \item $f:\R_{\geq 1}\to\R$ with $f(x)=\ln(x)$,
        \item $\tf:\F_{\geq 1}\to\F$ with $\tf(x)=\ln(x)e^\delta$ where $|\delta|\leq \e$,
        \item $b:\F\times\R\rightharpoonup\R$ with $b(x,y)=e^y$.
    \end{itemize}
\end{restatable}

With our new definition of composable backward stability, we have more lenses with which to construct backward stable programs. In our system, every subprogram satisfies the relational lens definition (\zcref{def:shel_def}). By composition of relational lenses (\zcref{thm:lens_comp}), the full program is also a relational lens, and thus is backward stable (\zcref{thm:lens_sound}).

\begin{remark}
  The terminology \emph{backward error lens} refers to lenses in functional and bidirectional programming \cite{lenses}. While there are many variants, a lens usually consists of a \emph{forward map} that extracts or computes and a \emph{backward map} that updates the source, subject to some \emph{lens laws}. In our setting, our two forward maps model the ideal and floating-point computations, and the backward map propagates error on the output back to the inputs. The main difference between our relational lenses and prior non-expansive lenses \cite{bean} is in the lens laws: the requirements on the forward and backward maps.
\end{remark}

%% file: figures/compose.tex
\begin{tikzcd}[cramped]
	x && {\tilde{f}(x)} && {\tilde{g}(\tilde{f}(x))} \\
	{\exists\tilde{x}} && {\exists \tilde{y}} \\
	{??}
	\arrow[from=1-1, to=1-3]
	\arrow[from=1-3, to=1-5]
	\arrow["f"{description}, from=2-1, to=1-3]
	\arrow["g"{description}, from=2-3, to=1-5]
	\arrow["f"{description}, from=3-1, to=2-3]
\end{tikzcd}

%% file: sections/category.tex
Our new definition of backward error lens is closed under composition, but it is tedious to construct backward error lenses for larger programs solely via composition. To address this challenge, we present a novel category $\Shel$ of \emph{shifted error lenses} where morphisms are relational lenses, and objects represent a set of possible values, along with a collection of possible ways to \emph{shift} or \emph{perturb} a value. Like other categories of lenses, $\Shel$ has rich categorical structure offering a variety of ways to combine relational lenses beyond simple composition. Thus we can work within the category, freely combining morphisms using these structures, while ensuring that the result is a valid relational lens. To prove backward stability of a given program, it suffices to construct a corresponding morphism in $\Shel$; the flexibility of the category allows a rich variety of backward error analyses, even for the same program.

In this section, we define the objects and morphisms of the category $\Shel$, show that it supports several general categorical constructions, and define several useful morphisms. 
\ifarxiv 
Omitted proofs can be found in \zcref{app:category_proofs}.
\fi

\subsection{Objects}
We think of an object in $\Shel$ as representing a program variable.  An object consists of a set of values---the range of the variable---as well as a group of \emph{shifts}, which describe how to perturb each value. Technically, shifts perturb the variable via a group action, and the magnitude of the perturbation is captured by a size function on the shifts. Objects also carry a constant that bounds the size of shifts introduced by morphisms, similar to $\alpha$ in a backward error lens (\zcref{def:bean_def}).

\begin{definition}
    An object in $\Shel$ is a triple $(X,S,p)$ where
\begin{itemize}
    \item $X$ is a set,
    \item $(S,*,e)$ is an abelian group equipped with
    \begin{itemize}
        \item a right action on $X$, also denoted $x*s$, and
        \item a componentwise norm $\|\cdot\|:S\to \R_{\geq 0}^n$,
    \end{itemize}
    \item and $p\in\R_{\geq 0}^n$ is a bound.
\end{itemize}
We call $n$ the \emph{dimension} of $S$, and we think of it as the number of independent shifts in $S$. All inequalities and equalities involving the norm are understood componentwise. The norm satisfies 
\begin{itemize}
    \item $\|s\|=0$ if and only if $s=e$ (positive definiteness), and
    \item $\|s_1*s_2\|\leq \|s_1\|+\|s_2\|$ (triangle inequality).
\end{itemize}
\end{definition}

We abbreviate $(X,S,p)$ as $X$ when clear. By definition, the group action $*:X\times S\to X$ satisfies $(x* s_1)* s_2=x*(s_1* s_2)$ and $x* e=x$. We think of $x*s$ as ``$x$ shifted by $s$.''

\begin{example}
    A single real variable is represented by $(\R,\R^{+},p)$ for some $p\in\R_{\geq 0}$, where $\R^{+}$ is the group $(\R,+,0)$. In this case, $\R$ is both the range of the variable and the group of shifts, which has dimension one and norm $|\cdot|$ (absolute value). The right action is $x* s=xe^s$, recalling our rounding error model (\zcref{def:olver_model}). The bound $p$ constrains the size of the shift. We check that the group action is well-defined:
    \begin{align*}
        (x *s_1)*s_2 = (xe^{s_1})e^{s_2} = xe^{s_1+s_2} = x*(s_1+s_2) \quad\text{and} \quad x*0 = xe^0 = x.
    \end{align*}
\end{example}

In our concrete instantiation of $\Shel$, the range of our variables and the groups of shifts will always be the reals. However, we work with a general presentation which may be instantiated with complex numbers or rotational shifts, for example.

\subsection{Morphisms}
Morphisms are a generalization of relational backward error lenses. A morphism is a triple $(f,\tf,b)$ where, as in \zcref{def:shel_def}, $f$ and $\tf$ are exact and approximate maps and $b$ is the backward map. There is an important difference: in a relational lens, the backward map directly takes a value $y$ in the output space and returns an $\tx$.  In $\Shel$, $b$ takes a \emph{shift} on the output and returns a \emph{shift} on the input---the value $y$ is implicitly described as a shift of $\tf(x)$, and the perturbed input $\tx$ is described by a shift of $x$. While this approach may seem rather indirect, we will soon see that these shifts enable our category to precisely capture variables that must be perturbed in some correlated way in order to prove backward stability. Formally, morphisms in $\Shel$ are defined as follows.

\begin{definition} \label{def:shel_morph}
    A morphism $(f,\tf,b):(X,S,p)\to (Y,T,q)$ is a triple of functions $f:X\to Y$, $\tf:X\rightharpoonup Y$, and $b:T\to S$ that satisfy the following. For all $x\in \dom(\tf)$ and $t\in T$ such that $\|t\|\leq q$,
    \begin{equation*}
        (1)\quad f(x* b(t))=\tf(x)*t \quad\text{and}\quad (2) \quad \|b(t)\|\leq p.
    \end{equation*}
\end{definition}

To develop some intuition for this definition, let's see how these morphisms generalize relational backward error lenses (\zcref{def:shel_def}). Setting $y=\tf(x)*t$ and $\tx=x*b(t)$, we can think of $y$ as $\tf(x)$ shifted by $t$, and $\tx$ as $x$ shifted by $b(t)$. The backward map encodes how a shift on $\tf(x)$ propagates back to a shift on $x$. Under this reading, the precondition $\|t\|\leq q$ says that $y$ is close to $\tf(x)$, corresponding to $d(\tf(x),y)\leq q$. The first condition is exactly $f(\tx)=y$ and the second becomes $d(x,\tx)\leq p$, recovering the definition of relational lens. Thus, in this category, we use the terms morphism and lens interchangeably.

For some more intuition, consider the identity shift $e\in T$ with no effect on the output $\tf(x)$. Since $\|e\|=0\leq q$, by the conditions we can conclude
\begin{equation*}
    (1) \quad f(x*b(e))=\tf(x)*e=\tf(x) \quad\text{and}\quad (2) \quad\|b(e)\|\leq p ,
\end{equation*}
so $x*b(e)$ gives us back the standard backward stability witness (\zcref{def:std_def}). 

\begin{example}
    We construct a morphism for the square root lens from \zcref{ex:sqrt_lens}. For any $p\in\R_{\geq 0}$, we define the morphism $\sqr:(\R,\R^+,2p+2\e)\to (\R,\R^+,p)$ with maps
    \begin{itemize}
        \item $f:\R\to\R$ is $f(x)=\sqrt{|x|}$,
        \item $\tf:\F\to\F$ is $\tf(x)=\sqrt{|x|}e^\delta$ for some $|\delta|\leq \e$ (from our model in \zcref{def:olver_model}),
        \item $b:\R\to\R$ is $b(s)=2s+2\delta$.
    \end{itemize}
\end{example}
When the computed output is unperturbed, the shift on the input is $b(0)=2\delta$, which we can think of as the error just from the square root computation. When the output is perturbed further by shift $s$, the resulting perturbation on $x$ itself is $2s$. By rewriting $y$ as a perturbation on $\tf(x)$ and $\tx$ as a perturbation on $x$ and working only with shifts, our proofs of the conditions simplify greatly. We now verify the conditions for these maps to form a $\Shel$ morphism.
\begin{proof}
    Take $x\in\F$ and $s\in\R$ such that $|s|\leq p$. Now,
    \begin{align*}
        (1)\quad f(x*b(s)) &= \sqrt{|xe^{2s+2\delta}|} = \sqrt{|x|}e^{s+\delta} & (2)\quad |b(s)| &= |2s+2\delta| \leq 2|s|+2|\delta| \\
        &= \tf(x)*s & &\leq 2p+2\e,
    \end{align*}
    satisfying the conditions.
\end{proof}

\subsection{Composition and Identity}
Just like relational lenses, morphisms in $\Shel$ compose.
\begin{restatable}{theorem}{morphismComp}
    Given morphisms $(f_1,\tf_1,b_1):(X,S,p)\to (Y,T,q)$ and $(f_2,\tf_2,b_2):(Y,T,q)\to (Z,U,r)$, their composition
    \begin{equation*}
        (f,\tf,b)=(f_2\circ f_1,\tilde{f}_2\circ \tilde{f}_1,b_1\circ b_2)
    \end{equation*}
    is a morphism $(X,S,p)\to (Z,U,r)$. Moreover, composition is associative.
\end{restatable}

Completing the basic category, the identity morphism is $(\id_X,\id_X,\id_S):(X,S,p)\to (X,S,p)$. The lens conditions are straightforward: for $x\in X$ and $s\in S$ such that $\|s\|\leq p$, we have
\begin{align*}
    (1) \quad \id_X(x*b(s))&=x*s=\id_X(x)*s \qquad (2) \quad \|b(s)\|=\|s\|\leq p.
\end{align*}

\begin{theorem}
    $\Shel$ forms a category with objects $(X,S,p)$ and morphisms $(f,\tf,b)$ generalizing the properties of a relational backward error lens.
\end{theorem}

\subsection{Tensor Product}
The tensor product provides one basic way to combine two objects; we will often think of the tensor product of a list of objects as a \emph{context} of variables. Intuitively, values of the tensor product are pairs of values from the two objects, and the shifts of the tensor product are pairs of shifts, each acting on the respective component.

\begin{definition}[Tensor Product]
The tensor product of two objects is
\begin{equation*}
    (X_1,S_1,p_1)\otimes (X_2,S_2,p_2)=(X_1\times X_2,S_1\times S_2,(p_1,p_2))
\end{equation*}
where $S_1\times S_2$ is the direct product of groups with dimension $n_1+n_2$. We define $\|(s_1,s_2)\|=(\|s_1\|,\|s_2\|)$. The right action of $S_1\times S_2$ on $X_1\times X_2$ is 
\begin{equation*}
    (x_1,x_2)* (s_1,s_2)=(x_1*s_1,x_2*s_2).
\end{equation*}
\end{definition}

Intuitively, the tensor product combines objects independently---the group action allows the two components to be perturbed with different shifts, possibly drawn from different groups. Given two morphisms $f_1:X_1\to Y_1$ and $f_2:X_2\to Y_2$, we can define a morphism $f_1\otimes f_2:X_1\otimes X_2\to Y_1\otimes Y_2$ to act in parallel on $X_1$ and $X_2$. For example, given the tensor product of three objects representing variables
\begin{equation*}
    (\R,\R^+,p_1)\otimes (\R,\R^+,p_2)\otimes (\R,\R^+,p_3),
\end{equation*}
we may perform $\add$ on the first two and $\sqr$ on the third in parallel, concretely forming a morphism with two outputs and approximating program
\begin{equation*}
    \tf(x,y,z)=(\add~x~y,\sqr~z).
\end{equation*}
In practice, it is common for all but one of the parallel morphisms to be the identity, thus applying morphisms one at a time. We can show that $\otimes$ is a bifunctor and that we have a unit, associators, and unitors, allowing us to reassociate variables in a context. Moreover, the tensor product is symmetric, so we may rearrange tensored variables in the context as we wish.
\begin{restatable}{theorem}{symMon}\label{lem:sym_mon}
    $\Shel$ is a symmetric monoidal category with tensor product $\otimes$ and unit $I=(\{\diamond\}, \{e\},0)$.
\end{restatable}

\subsection{Share Product}
As is common in lens categories, the tensor product is not a Cartesian product. For instance, we cannot define a diagonal lens $X\to X\otimes X$, as the unrestricted duplication of variables may not maintain backward stability. We have a restricted form of duplication, however: variables may be copied as long as they are then perturbed \emph{identically}. Concretely, we can define the following duplication lens.

\begin{lemma}\label{lem:duplicator}
    For any object $(X,S,p)$, we have a morphism
    \begin{equation*}
        \Delta:(X,S,p)\to (X\times X, S, p)
    \end{equation*}
    where the action in the target object is $(x_1,x_2)*s=(x_1*s,x_2*s)$.
\end{lemma} 
This group action ensures that the same shift is applied to the two copies of the duplicated variable, so they can always be reconciled into a single perturbed variable. This lens is a prime example of how group actions in our category can describe correlated perturbations, enabling more flexible use of variables while guaranteeing backward stability.
\begin{proof}
    We define $\Delta(x)=\tilde{\Delta}(x)=(x,x)$ and $b(s)=s$. Let us prove the lens conditions for $x\in X$ and $s\in S$ such that $\|s\|\leq p$:
    \begin{equation*}
        (1)\quad \Delta(x*s) = (x*s,x*s) =(x,x)*s=\tilde{\Delta}(x)*s \qquad
        (2)\quad \|b(s)\|=\|s\|\leq p,
    \end{equation*}
    satisfying the conditions.
\end{proof}

We can generalize this idea to any two objects with the same shift group:
\begin{equation}\label{eq:share_tens}
    \share: (X_1,S,p)\otimes (X_2,S,p)\to (X_1\times X_2, S,p).
\end{equation}
This lens takes two objects to their \emph{share product}, bonding them so they are always perturbed identically. We can use the share product to define lenses for $\add$, $\sub$, and $\mul$ that distribute error equally onto their operands. 

\begin{restatable}{example}{addSubMul}
    We have morphisms
    \begin{align*}
        \add&:(\R\times\R,\R^+,p+\e)\to (\R,\R^+,p) \\
        \sub&:(\R\times\R,\R^+,p+\e)\to (\R,\R^+,p)\\
        \mul&:(\R\times\R,\R^+,(p+\e)/2)\to (\R,\R^+,p).
    \end{align*}
\end{restatable}
\begin{proof}
    We detail the lens for $\sub$. For any $p\geq 0$, we define 
    \begin{equation*}
        f(x,y)=x-y \qquad \tf(x,y)=\sub~x~y=(x-y)e^\delta \qquad b(s)=s+\delta.
    \end{equation*}
    For $(x,y)\in \F\times \F$ and $s\in\R$ such that $|s|\leq p$,
    \begin{align*}
        (1) \quad f((x,y)*b(s)) &= f(x*(s+\delta),y*(s+\delta)) & (2) \quad |b(s)| &=|s+\delta| \leq |s|+|\delta| \\
        &= xe^{s+\delta}-ye^{s+\delta} =\tf(x,y)*s &&\leq p+\e,
    \end{align*}
    satisfying the conditions.
\end{proof}

\subsection{Push Product}
We now define a new construct called the \emph{push product} that expresses a different relationship between variables than the share product does. It is motivated by a program with two subprograms $x$ and $\mul~x~y$: we may compute with $\mul~x~y$ and push all the error onto $y$ without affecting $x$, but perturbing $x$ necessarily perturbs $\mul~x~y$. The push product, a generalization of the tensor product, captures this relationship, which is parametrized by a function $i$ that tells us how perturbations on the first object are pushed onto the second object.

\begin{definition}[Push Product]
Given a group homomorphism $i:S_1\to S_2$, the push product $X_1\star_iX_2$ of two objects is 
\begin{equation*}
    (X_1,S_1,p_1)\star_i (X_2,S_2,p_2)=(X_1\times X_2, S_1\times S_2, (p_1,p_2))
\end{equation*}
just like the tensor product. However, the right action of $S_1\times S_2$ on $X_1\times X_2$ is now
\begin{equation*}
    (x_1,x_2)*(s_1,s_2)=(x_1*s_1,x_2*i(s_1)*s_2). 
\end{equation*}
The second object in a push product is said to be \emph{dependent} on the first.
\end{definition}

Intuitively, the group action enforces that if we shift $x_1$ by $s_1$, we must shift $x_2$ by at least $i(s_1)$, plus a possible additional shift $s_2$. The tensor product is a special case of the push product when we use the trivial homomorphism $\varphi_e$ that maps everything to the identity. For an object $X_1\star_{\varphi_e}X_2$, when we shift $x_1$ by $s_1$, we shift $x_2$ by the identity, that is to say, not shift $x_2$ at all. Hence, we can define the tensor product in terms of the push product:
\begin{equation*}
    X_1\otimes X_2=X_1\star_{\varphi_e} X_2.
\end{equation*}
Using the push product, we can define a multiplication lens that pushes all the error onto its second argument while returning its first argument for further use, as in \zcref{eq:dmul}.

\begin{example}
    We define the lens $\dmul:(\R,\R^+,p)\otimes(\R,\R^+,q+\e)\to (\R,\R^+,p)\star_{\id}(\R,\R^+,q)$ by 
    \begin{equation*}
        f(x,y)=(x,xy) \qquad \tf(x,y)=(x,xye^\delta) \qquad b(s,t)=(s,t+\delta).
    \end{equation*}
\end{example}
\begin{proof}
    We use $*_1$ to denote the group action in the source object (the tensor product) and $*_2$ for the target (the push product). For $(x,y)\in\F\times \F$ and $(s,t)\in \R\times\R$ such that $|s|\leq p$ and $|t|\leq q$,
    \begin{align*}
        (1)\quad &f((x,y)*_1(s,t+\delta)) &(2)\quad \|b(s,t)\| &= \|(s,t+\delta)\| \\ 
        &= f(xe^s,ye^{t+\delta}) = (xe^s, xye^{s+t+\delta}) & &=(|s|,|t+\delta|) \\
        &= (x,xye^\delta)*_2(s,t)= \tf(x,y)*_2(s,t) & &\leq (p,q+\e),
    \end{align*}
    satisfying the conditions.
\end{proof}
This lens multiplies two independent variables $x$ and $y$ and asserts a dependency of the returned value $\mul~x~y$ on $x$ using the push product. The backward map specifies that the error $\delta$ from the multiplication is pushed entirely onto the second variable $y$. We call this lens $\dmul$ to differentiate it from the multiplication lens $\mul$, which distributes the error equally onto both operands.

The following lens asserts that if the variable $y$ is already dependent on $x$, then their product has an even stronger dependency on $x$.
\begin{example}\label{ex:dmul_lens}
    With similar proof, for any $n\in\N$, we can generalize this lens further to 
    \begin{equation*}
        \dmul:(\R,\R^+,p)\star_i(\R,\R^+,q+\e)\to (\R,\R^+,p)\star_j(\R,\R^+,q)
    \end{equation*}
    where $i(s)=n\cdot s$ and $j(s)=(n+1)\cdot s$.
\end{example}

We now define a few lenses which are useful in an analysis to manipulate a context involving the push product. First, we can modify the homomorphism $i$ as we wish by decreasing the error bound on the dependent object.

\begin{restatable}{lemma}{pushAdjust}\label{lem:push_adjust}
    We have a morphism 
    \begin{equation*}
    \push:(X_1,S_1,p_1)\star_i(X_2,S_2,p+p_2)\to (X_1,S_1,p_1)\star_j(X_2,S_2,p_2)
    \end{equation*}
    where $p=\max_{\|s_1\|\leq p_1}\|i(s_1)^{-1}* j(s_1)\|$.
\end{restatable}
Moreover, we have projections out of the push product, and a distributive morphism.
\begin{restatable}{lemma}{projections}
    For objects $X_1$ and $X_2$, we have morphisms
    \begin{equation*}
        \pi_1:X_1\star_iX_2\to X_1\quad\text{and}\quad \pi_2:X_1\star_iX_2\to X_2.
    \end{equation*}
\end{restatable}

\begin{restatable}{lemma}{distributor}\label{lem:distributor}
    For objects $X_1,X_2$ and $Y_1,Y_2$, we have a morphism 
    \begin{equation*}
        \Theta:(X_1\star_i Y_1)\otimes(X_2\star_j Y_2) \to (X_1\otimes X_2)\star_{(i,j)}(Y_1 \otimes Y_2)
    \end{equation*}
    where $(i,j)$ is the parallel application of homomorphisms.
\end{restatable}

While it is in general not possible to compose maps in parallel with a push product, parallel composition is possible under certain conditions
\ifarxiv
detailed in \zcref{lem:push_parallel}.
\else
\cite{supplement}.
\fi

\begin{example}
    We can apply the $\add$ lens to the dependent in a push product with a lens
    \begin{equation*}
        \id\star\ \add:(\R,\R^+,p)\star_i(\R\times\R,\R^+,q+\e)\to (\R,\R^+,p)\star_i(\R,\R^+,q)
    \end{equation*}
    whose maps are 
    \begin{equation*}
        f(a,(b,c))=(a,b+c) \qquad \tf(a,(b,c))=(a,(b+c)e^\delta) \qquad b(s,t)=(s,t+\delta).
    \end{equation*}
\end{example}

Finally, we can generalize the $\share$ lens (\zcref{eq:share_tens}) to convert
a push product with the same shift groups into a share product.
\begin{restatable}{lemma}{shareLens}\label{lem:share}
    Given any two objects with the same group of shifts, we can combine them into the share product with a lens:
    \begin{equation*}
        \share:(X_1,S,p_1)\star_i(X_2,S,p_2)\to (X_1\times X_2,S,p_1)
    \end{equation*}
    where $p_2=\max_{\|s\|\leq p_1}\|i(s)^{-1}*s\|$.
\end{restatable}
Intuitively, the lenses $\push$, $\share$, $\Delta$, $\pi$, and $\Theta$ represent trivial computations---their underlying exact and approximate maps are in most cases the identity. However, they are useful to model the manipulation of variables and error terms often required in a backward error analysis. As we will see in the next section, when we build complex backward stable programs in the category, we will often use these lenses freely until the context is in the correct shape to apply a ``computational'' lens for some arithmetic operation.

%% file: sections/case_studies.tex
In this section, we show how to construct $\Shel$ lenses witnessing backward stability for example programs. For readability, we elide lenses that permute and reassociate the tensor product.
\begin{example}
    To warm up and introduce our notation, consider the program
    \begin{equation*}
        \tf(x,y)=\add~x~(\mul~x~y) \quad \text{implementing} \quad f(x,y)=x+xy.
    \end{equation*}
    In $\Bean$~\citep{bean}, this program is disallowed because it reuses the variable $x$, and the add operation pushes backward error onto $x$. But it is backward stable with error bounds of $\e$ for both $x$ and $y$, and we can form a lens to prove this in $\Shel$. We interpret the input as the object
    \begin{equation*}
        (\R,\R^+,\e)\otimes (\R,\R^+,\e),
    \end{equation*}
    i.e. two variables $x$ and $y$ with backward error bounds of $\e$ each. Next, we apply the $\dmul$ lens: 
    \begin{align*}
        \dmul:(\R,\R^+,\e)\otimes(\R,\R^+,\e)&\to (\R,\R^+,\e)\star_1(\R,\R^+,0)
    \end{align*}
    with underlying program $(x,y)\mapsto(x,\mul~x~y)$; this lens threads the
    input variable $x$ through to the output, so it can be reused in later computations. Here, $\star_n$ is the push product with homomorphism $s\mapsto n\cdot s$. Now, we need to add our two subprograms. To get the context in the right shape to apply the $\add$ lens, we use the lens
    \begin{equation*}
        \share:(\R,\R^+,\e)\star_1(\R,\R^+,0)\to (\R\times\R,\R^+,\e)
    \end{equation*}
    from \zcref{lem:share}, which corresponds to the identity $(x,z)\mapsto (x,z)$. Finally, we use the $\add$ lens:
    \begin{equation*}
        \add:(\R\times\R,\R^+,\e)\to(\R,\R^+,0)
    \end{equation*}
    corresponding to the program $(x,z)\mapsto \add~x~z$. Composing these three lenses gives us a final lens
    \begin{equation*}
        (f,\tf,b):(\R,\R^+,\e)\otimes(\R,\R^+,\e)\to(\R,\R^+,0)
    \end{equation*}
    with our desired maps $f,\tf$. We can write this lens concisely as follows:
    \begin{align*}
        \denot{x}_\e\otimes\denot{y}_\e&\fixedarrow{\dmul}\denot{x}_\e~ \star_1~\denot{\mul~x~y}_0 \\
        &\fixedarrow{\share}\denot{x,~\mul~x~y}_\e \\
        &\fixedarrow{\add}\denot{\add~x~(\mul~x~y)}_0
    \end{align*}
    where each context is an object and each arrow is a morphism.  In this notation, $\denot{e_1,\dots,e_n}_p$ represents the object $(\R^n,\R^+,p)$ and the expressions $e_k$ describe the approximate maps that have been applied to reach this object.

    Now, we can conclude a backward error guarantee for $\tf$, reading off the per-variable backward error from the source context $\denot{x}_\e\otimes \denot{y}_\e$, using the two conditions on $\Shel$ morphisms (\zcref{def:shel_morph}). For floating-point inputs $x$ and $y$, by the first morphism condition,
    \begin{equation*}
        f((x,y)*b(0))=f(\tx,\ty)=\tf(x,y)*0=\tf(x,y).
    \end{equation*}
    By the second, $\|b(0)\|=(d(x,\tx),d(y,\ty))\leq (\e,\e)$, showing that the distance between the inputs and perturbed inputs is at most $\e$ each, and thus our program is backward stable.
\end{example}

\begin{example}
  To demonstrate the square root lens, we perform the Cholesky factorization of
  a $2\times 2$ symmetric positive definite (SPD) matrix, which computes a lower triangular matrix $L$ such that $A=LL^\top$:
    \begin{equation*}
        \begin{bmatrix}
            a_{11} & a_{21} \\
            a_{21} & a_{22}
        \end{bmatrix} =
        \begin{bmatrix}
            \ell_{11} & 0 \\
            \ell_{21} & \ell_{22}
        \end{bmatrix} 
        \begin{bmatrix}
            \ell_{11} & \ell_{21} \\
            0 & \ell_{22}
        \end{bmatrix}.
    \end{equation*}
    The entries for $L$ are computed as follows:
    \begin{equation*}
        \ell_{11} = \sqrt{a_{11}} \qquad\qquad
        \ell_{21} = \frac{a_{21}}{\ell_{11}} \qquad\qquad
        \ell_{22} = \sqrt{a_{22}-\ell_{21}^2}.
    \end{equation*}
    Beginning with the context of input variables, we first compute $\ell_{11}$:
    \begin{equation*}
        \denot{a_{11}}_{2\e} \otimes \denot{a_{21}}_{3\e} \otimes \denot{a_{22}}_{3\e}\fixedarrow{\sqr~\otimes~\id~\otimes~\id}\denot{\underbrace{\sqr~a_{11}}_{\ell_{11}}}_{0} \otimes \denot{a_{21}}_{3\e} \otimes \denot{a_{22}}_{3\e}
    \end{equation*}
    Then, to compute $\ell_{21}$, we use a division lens:
    \begin{equation*}
        \divi:(\R,\R^+,p+\e)\otimes (\R,\R^+,0)\to (\R,\R^+,p)
    \end{equation*}
    which pushes all the error onto the numerator, defined in \ifarxiv\zcref{lem:div_lens}\else\cite{supplement}\fi.\footnote{For simplicity, we define the lens to return $0$ if the denominator is $0$, but if $A$ is SPD then this will not happen.} Our context is
    \begin{equation*}
        \denot{\ell_{11}}_{0} \otimes \denot{a_{21}}_{3\e} \otimes
        \denot{a_{22}}_{3\e}\fixedarrow{\divi~\otimes~\id}\denot{\underbrace{\divi~a_{21}~\ell_{11}}_{\ell_{21}}}_{2\e} \otimes \denot{a_{22}}_{3\e} .
    \end{equation*}
    The last element, $\ell_{22}$, is computed by:
    \begin{align*}
        \denot{\ell_{21}}_{2\e} \otimes \denot{a_{22}}_{3\e} 
        &\fixedarrow{\Delta~\otimes~\id} \denot{\ell_{21},~\ell_{21}}_{2\e} \otimes \denot{a_{22}}_{3\e} \tag{\zcref{lem:duplicator}} \\
        &\fixedarrow{\mul~\otimes~\id}\denot{\mul~\ell_{21}~\ell_{21}}_{3\e} \otimes \denot{a_{22}}_{3\e} \\
        &\fixedarrow{\share}\denot{a_{22},~\mul~\ell_{21}~\ell_{21}}_{3\e} \tag{\zcref{lem:share}} \\
        &\fixedarrow{\sub} \denot{\sub~a_{22}~(\mul~\ell_{21}~\ell_{21})}_{2\e} \\
        &\fixedarrow{\sqr}\denot{\sqr~(\sub~a_{22}~(\mul~\ell_{21}~\ell_{21}))}_0
    \end{align*}
    where $\Delta$ is the diagonal map and $\share$ is the share product map. By inspecting the annotations on the input, we conclude that computing $\ell_{22}$ is backward stable with a backward error bound of $2\e$ on $a_{11}$ and $3\e$ on $a_{21}$ and $a_{22}$. We can construct similar lenses to prove backward stability of computing the other entries $\ell_{11}$ and $\ell_{21}$.\footnote{We note that these lenses compute a single entry of $L$, rather than the entire matrix $L$---proving backward stability of the matrix version requires the SPD property of $A$, which seems difficult to encode directly in $\Shel$.}
\end{example}

\begin{example}
    Consider the polynomial $f(a,x)=x+ax^2$ and program 
    \begin{equation*}
        \tf(a,x)=\add~x~(\mul~(\mul~a~x)~x).
    \end{equation*}
    A backward stability result for this program needs to show
    \begin{align*}
        \tf(a,x)&= (x+ax^2e^{\delta_1+\delta_2})e^{\delta_3} \\
        &= xe^{\delta_3}+ax^2e^{\delta_1+\delta_2+\delta_3} \\
        &= \tx+\tilde{a}\tx^2 =f(\tilde{a},\tilde{x})
    \end{align*}
    for some $\tilde{a}$ and $\tilde{x}$. This requires that we pick
    $\tx=xe^{\delta_3}$. To figure out how to pick $\tilde{a}$, we need
    \begin{align*}
        ax^2e^{\delta_1+\delta_2+\delta_3} &= \tilde{a}\tx^2 = \tilde{a}(xe^{\delta_3})^2 =\tilde{a}x^2e^{2\delta_3},
    \end{align*}
    so we must set $\tilde{a}=ae^{\delta_1+\delta_2-\delta_3}$; note that by the triangle inequality, the backward error on $a$ is $d(a,\tilde{a})=|\delta_1+\delta_2-\delta_3|\leq 3\e$. Intuitively, this analysis compensates for the squared error from the $x^2$ term by pushing the \emph{inverse} of the error onto $a$. This analysis is possible in $\Shel$, relying on the fact that the set of shifts for an object forms a \emph{group}, with each shift having an inverse. More concretely, we can construct the following lens in $\Shel$: 
    \begin{align*}
        \denot{a}_{3\e} \otimes \denot{x}_{\e} &\fixedarrow{\dmul} \denot{x}_{\e} ~\star_1~\denot{\mul~a~x}_{2\e} \\
        &\fixedarrow{\dmul} \denot{x}_{\e} ~\star_2~ \denot{\mul~(\mul~a~x)~x}_{\e} \\
        &\fixedarrow{\share} \denot{x,~\mul~(\mul~a~x)~x}_\e \tag{\zcref{lem:share}} \\
        &\fixedarrow{\add} \denot{\add~x~(\mul~(\mul~a~x)~x)}_0.
    \end{align*}
    Recall that the backward map of $\share$ here is $b(s)=(s,s^{-1})$; the inverse shift corresponds to the $-\delta_3$ shift that we saw above. We conclude that the program is backward stable with a backward error bound of $3\e$ on $a$, and $\e$ on $x$.
\end{example}

\begin{example}
    For positive weights $w_1$, $w_2$ and values $x_1$, $x_2$, we can compute a weighted average: 
    \begin{equation*}
        f(w,x)=\frac{w_1x_1+w_2x_2}{w_1+w_2}.
    \end{equation*}
    The following lenses compose to form this program:
    \begin{align*}
        &\denot{w_1}_{\e} \otimes \denot{w_2}_{\e} \otimes \denot{x_1}_{4\e} \otimes \denot{x_2}_{4\e} \\ &\fixedarrow{\dmul~\otimes~\dmul} (\denot{w_1}_\e ~\star_1~ \denot{\mul~w_1~x_1}_{3\e}) \otimes (\denot{w_2}_{\e}~\star_1~\denot{\mul~w_2~x_2}_{3\e}) \\
        &\fixedarrow{\push~\otimes~\push} \denot{w_1}_\e \otimes \denot{\mul~w_1~x_1}_{2\e} \otimes \denot{w_2}_{\e}\otimes \denot{\mul~w_2~x_2}_{2\e} \tag{\zcref{lem:push_adjust}}\\
        &\fixedarrow{\share~\otimes~\share} \denot{w_1,~w_2}_{\e} \otimes \denot{\mul~w_1~x_1,~\mul~w_2~x_2}_{2\e} \tag{\zcref{lem:share}} \\
        &\fixedarrow{\add~\otimes~\add}  \denot{\add~w_1~w_2}_{0} \otimes \denot{\add~(\mul~w_1~x_1)~(\mul~w_2~x_2)}_{\e} \\
        &\fixedarrow{\divi} \denot{\divi~(\add~(\mul~w_1~x_1)~(\mul~w_2~x_2))~(\add~w_1~w_2)}_0
    \end{align*}
    where the $\push$ lens adjusts the push product homomorphism and the $\share$ lens constructs the share product. The lens shows that the program is backward stable with backward error bounds of $4\e$ on each of the values and $\e$ on the weights. We could have produced a tighter analysis by applying more specialized lenses instead of the $\push$ lens, which frees the variables from under the push product. For example,
    \begin{align*}
        &\denot{w_1}_{0}\otimes \denot{w_2}_{0}\otimes \denot{x_1}_{4\e}\otimes\denot{x_2}_{4\e} \\
        &\fixedarrow{\dmul~\otimes~\dmul} (\denot{w_1}_0 ~\star_1~ \denot{\mul~w_1~x_1}_{3\e}) \otimes (\denot{w_2}_{0}~\star_1~\denot{\mul~w_2~x_2}_{3\e}) \\
        &\fixedarrow{\Theta} (\denot{w_1}_0\otimes \denot{w_2}_0)\star_{(1,1)}(\denot{\mul~w_1~x_1}_{3\e}\otimes\denot{\mul~w_2~x_2}_{3\e}) \tag{\zcref{lem:distributor}}\\
        &\fixedarrow{\share~\star~\share} \denot{w_1,~w_2}_0~\star_{1}~\denot{\mul~w_1~x_1,~\mul~w_2~x_2}_{3\e} \\
        &\fixedarrow{\id~\star~\add} \denot{w_1,~w_2}_0~\star_{1}~\denot{\add~(\mul~w_1~x_1)~(\mul~w_2~x_2)}_{2\e} \\
        &\fixedarrow{\mathbf{adddiv}}\denot{\divi~(\add~(\mul~w_1~x_1)~(\mul~w_2~x_2))~(\add~w_1~w_2)}_0,
    \end{align*}
    where $\Theta$ is the push product distributor and $\mathbf{adddiv}$ is a lens which performs addition then division, pushing the error from both computations onto the numerator (defined in \ifarxiv\zcref{lem:add_div_lens}\else\cite{supplement}\fi). This alternative analysis is tighter---since it pushes no error onto the weights---although it is a bit more complicated and requires another primitive lens.
\end{example}

%% file: sections/implementation.tex
Having constructed a solid categorical foundation for sound backward error analysis, we now use it to guide an automated error analysis. Considering the lenses applied to analyze a program in \zcref{sec:case_studies}, we note that lenses can be roughly classified as \emph{computational} (such as $\add$) or \emph{structural} (such as $\pi_1$ and $\share$). The computational lenses are determined by the syntax of the program, but choosing the right structural lenses to apply between computational steps is nontrivial. At a high level, we design an automated method to search over possible structural lenses to build a lens corresponding to a given target program.

In more detail, our algorithm takes in a symbolic expression $e$ with $n$ variables and returns whether a lens was found of the shape
\begin{equation*}
    (\R,\R^+,p_1)\otimes\dots\otimes (\R,\R^+,p_n)\to (\R,\R^+,0)
\end{equation*}
with underlying maps corresponding to the program $e$. If we can find such a lens, then we conclude a backward error guarantee as in \zcref{sec:case_studies}. Our algorithm starts from the target object $(\R,\R^+,0)$ and works backward, applying maps corresponding to the program's symbolic expression in reverse.

We implement our algorithm in $\egglog$, a reasoning system combining logic programming and equality saturation \cite{egglog}. Finding lenses requires non-deterministic search over possible lenses at each step (as there are often many choices), and rewriting equivalent mathematical expressions and isomorphic variable contexts. With its combination of Datalog-style proof search and equality saturation, $\egglog$ is a natural solver for our analysis. We call our implementation $\eggshel$.

\subsection{A Syntactic Representation}
An immediate challenge with searching for morphisms in $\Shel$ is that the category is very general and abstract. To narrow our search, we develop a syntax representing a useful subcategory of $\Shel$, which we then encode in $\eggshel$. First, we consider expressions of the form
\begin{equation*}
    e ::= x\in\Var \mid \mathbf{unop}~e\mid \mathbf{binop}~e_1~e_2
\end{equation*}
where $\mathbf{unop}\in \{\sqr\}$, $\mathbf{binop}\in\{\add, \mul\}$. An expression $e$ represents a numerical program with variables as inputs and a single output. To keep the prototype simple, we do not implement all potential operations, though subtraction and division could easily be added. Next, we define a context $\Gamma$ in three layers as
\begin{align*}
    B &::= (x_1,\dots,x_n:(i, \R^n, p)), \quad x_k\in\Var,i\in\N,n\in\N,p\in\R_{\geq 0} \\
    T &::= B\star(B_1\otimes \dots\otimes B_m),\quad m\geq 0 \\
    \Gamma &::= \cdot \mid T\otimes \Gamma
\end{align*}
We assume that all variables are distinct. A context denotes the tensor and push product of several $\Shel$ objects. A \emph{base} $B$ represents an object $(\R^n,\R^+,p\e)$ along with $n$ bound variables $x_1,\dots,x_n$. A \emph{tree} $T$ represents the push product of one base $B$ onto $m$ dependent bases. For each dependent base in a tree, the constant $i$ specifies the homomorphism onto that object; for the root base of a tree, $i$ is ignored. In $\Shel$, push products can be formed between arbitrary objects, but we restrict our search to trees of height one for simplicity. Finally, a context $\Gamma$ is the tensor product of several trees. Formally, we inductively interpret each context $\Gamma$ as a $\Shel$ object: 
\begin{gather*}
    \denot{B} = (\R^n, \R^+,p\e) \qquad
    \denot{T} =\denot{B}\star_{i}(\denot{B_1}\otimes\dots\otimes \denot{B_m}), \quad i(s)=(i_1\cdot s,\dots,i_m\cdot s) \\
    \denot{\Gamma} = I \mid \denot{T}\otimes\denot{\Gamma}, \quad I=\text{unit object}
\end{gather*}
Corresponding to the shape of these contexts, we form a tuple of expressions $E$ as 
\begin{align*}
    E_B ::= (e_1,\dots,e_n) \qquad
    E_T ::= (E_B; {E_B}_1,\dots,{E_B}_m) \qquad
    E ::= \cdot \mid (E_T,E)
\end{align*}

Now, we define a relation $\RR$ on two contexts and a tuple of expressions. Intuitively, if $\RR(\Gamma_1,E,\Gamma_2)$, then there is a $\Shel$ morphism from $\denot{\Gamma_1}\to\denot{\Gamma_2}$ whose ideal and approximate maps correspond to $E$. Concretely, for each expression $e$, we can define its ideal and floating-point semantics $(f,\tf)$ with $f=\denot{e}_{\text{re}}$ and $\tf=\denot{e}_{\text{fp}}$, where:
\begin{align*}
    \denot{x}_{\text{re}} &= \pi ~\text{(projection)} & \denot{x}_{\text{fp}} &= \pi \\
    \denot{\mathbf{unop}~e}_{\text{re}} &= unop \circ\denot{e}_{\text{re}} & \denot{\mathbf{unop}~e}_{\text{fp}} &= \mathbf{unop}\circ \denot{e}_{\text{fp}} \\
    \denot{\mathbf{binop}~e_1~e_2}_{\text{re}} &= binop\circ (\denot{e_1}_{\text{re}},\denot{e_2}_{\text{re}}) & \denot{\mathbf{binop}~e_1~e_2}_{\text{fp}} &= \mathbf{binop}\circ (\denot{e_1}_{\text{fp}},\denot{e_2}_{\text{fp}})
\end{align*}
Tuples of expressions $E$ are interpreted as tuples of the corresponding functions. 

A prerequisite for $(\Gamma_1,E,\Gamma_2)$ to be in $\RR$ is that $E$ is a program which operates on variables in the context $\Gamma_1$ and outputs expressions corresponding to the variables of $\Gamma_2$. Thus, we require that the free variables of $E$ are a subset of the variables bound in $\Gamma_1$ (write $\FV(E)\subseteq\Var(\Gamma_1)$), and $E$ has the same shape as $\Gamma_2$.

\begin{figure}
  \centering
  \input{figures/relation}
  \caption{A subset of the rules defining $\RR(\Gamma_1,E,\Gamma_2)$.}
  \Description{A number of syntactic rules defining the relation $\RR$.}
  \label{fig:relation}
\end{figure}

The relation $\RR$ is defined inductively; we show a subset of the rules in \zcref{fig:relation}. Our implementation uses a slightly larger set of rules, which we found to be sufficient to analyze our example programs. Note that because we only implement a subset of the lenses available in $\Shel$, completeness (finding a lens for any backward stable program) is not guaranteed in $\eggshel$. For example, programs evaluating higher-order polynomials are not currently supported in $\eggshel$, though they are backward stable and do correspond to morphisms in $\Shel$.

Each rule corresponds to a lens introduced in \zcref{sec:category} applied to a concrete object. For each rule, we have several versions corresponding to where the lens may be applied in a context, such as in the multiple \textsc{Add} rules, but we have omitted these variations for brevity. Importantly, the \textsc{Trans} rule performs the syntactic substitution of the expressions in $E_1$ for the variables in $\Gamma_2$, which is straightforward since $E_1$ and $\Gamma_2$ have the same shape. Then $\RR(\Gamma_1,E_3,\Gamma_3)$ is well-defined since $\FV(E_3)=\FV(E_1)\subseteq \Var(\Gamma_1)$. With the relation in place, we can state and prove soundness.

\begin{theorem}[Soundness]
    If $\RR(\Gamma_1,E,\Gamma_2)$, then we have a morphism $(f,\tf,b):\denot{\Gamma_1}\to\denot{\Gamma_2}$ in $\Shel$ where $f=\denot{E}_{\text{re}}$ and $\tf=\denot{E}_{\text{fp}}$.
\end{theorem}
\begin{proof}
    By induction on the derivation of the relation $\RR$, and case analysis on the last rule applied. We display a few representative cases. 

    \textbf{Case (\textsc{Cong}).} This rule says that we can apply the relation to only some part of a context. Categorically, it corresponds to the congruence property of the tensor product. By our inductive hypothesis on the premise, we have a morphism $(f,\tf,b):\denot{\Gamma_2}\to\denot{\Gamma_3}$. We can apply it in parallel with $\id_{\denot{\Gamma_1}}$ to get a morphism
    \begin{equation*}
        \id_{\denot{\Gamma_1}}\otimes (f,\tf,b):\denot{\Gamma_1}\otimes\denot{\Gamma_2}\to\denot{\Gamma_1}\otimes\denot{\Gamma_3}
    \end{equation*}
    which corresponds to $(\id_{\Gamma_1},E)$. 
    
    \textbf{Case (\textsc{Share}).} This rule corresponds to the share product lens (\zcref{lem:share}), whose underlying maps are identity, in the specific instantiation of 
    \begin{equation*}
        \share: (\R^n,\R^+,p)\otimes (\R^m,\R^+,p)\to (\R^{n+m},\R^+,p).
    \end{equation*}
    
    \textbf{Case (\textsc{DMul}).} This rule corresponds to an instantiation of the $\dmul$ lens (\zcref{ex:dmul_lens}),
    \begin{equation*}
        \dmul:(\R,\R^+,p)\otimes (\R,\R^+,q+\e)\to (\R,\R^+,p)\star_1(\R,\R^+,q)
    \end{equation*}
    with underlying approximate map $\tf(x_1,x_2)=(x_1,\mul~x_1~x_2)$. \qedhere
\end{proof}

\subsection{Implementation}
$\eggshel$ encodes the relation $\RR$ in $\egglog$, a reasoning system which allows users to employ both logic programming and equality saturation \cite{egglog}. In $\egglog$, users define inductive datatypes, rewrite rules for equivalent terms, and rules to define relations on terms. Given a starting term, $\egglog$ derives as many facts as possible using these rules until saturation---when no more facts can be derived. For example, we define expression tuples $E$ and contexts $\Gamma$ as

\begin{lstlisting}
(datatype Expr            (datatype Ctx
    (Var String)              (Base f64 Expr f64) 
    (Vars Expr Expr)          (Tens Ctx Ctx)
    (Sqrt Expr)               (Star Ctx Ctx)) 
    (Add Expr Expr)
    (Mul Expr Expr))
\end{lstlisting}
Inductively, an expression tuple may be a single variable (a string), the square root, addition, or multiplication of expressions, or a tuple of expressions. A context is either a base object, the tensor of two contexts, or the star (push product) of one context onto another. We include some rewrite rules for equivalent expressions, such as this one expressing commutativity of
floating-point addition:
\begin{lstlisting}
(rewrite (Add e1 e2) (Add e2 e1))
\end{lstlisting}
Then, we declare a relation representing $\RR$ as a binary relation on contexts:
\begin{lstlisting}
(relation -> (Ctx Ctx))
\end{lstlisting}
To avoid performing large substitutions, we store and use expressions in place of variables within each context. For example, we can assert 
\begin{lstlisting}
(-> (Base 0.0 (Add  (Var "a") (Var "b")) 2.0)
    (Base 0.0 (Sqrt (Add (Var "a") (Var "b"))) 0.0))
\end{lstlisting}
corresponding to $\RR((x:(0,\R,2)),\sqr~{x},(y:(0,\R,0)))$.

Using this framework, we encode the rules for deriving $\RR$ from \zcref{fig:relation} as $\egglog$ rules. We also include some rewrite rules for contexts, such as associativity and commutativity for the tensor product, which correspond to the isomorphisms for these context rearrangements (\zcref{lem:sym_mon}).

Finally, in $\eggshel$, the user inputs a program like 
\begin{lstlisting}
(Add (Mul a a) (Mul b b))
\end{lstlisting}
and allows the relation to saturate to maximize the chance of finding the correct lens. The tool then checks whether a starting context with those variables was found in the relation, and outputs the per-variable error bounds for that context. If multiple backward error bounds were found, $\eggshel$ outputs all of them. 

\subsection{Evaluation}
We evaluated $\eggshel$ on several challenging examples of backward stable programs, many of which are not supported by previous backward error analysis tools. All experiments were performed on a MacBook Pro with an Apple M3 chip and 16 GB of memory. We briefly summarize our conclusions.

First, in \zcref{tab:first}, we run $\eggshel$ on five practical programs: \texttt{sum} naively sums a list of $n$ variables, \texttt{linear} evaluates a function $x+a_1x+\dots +a_nx$, \texttt{norm} takes the Euclidean norm of a length $n$ vector, \texttt{quad} evaluates a quadratic function $x+a_1x^2+\dots+a_nx^2$, and \texttt{dotprod} computes the dot product of two length $n$ vectors. We include a time comparison to $\Bean$ \cite{bean} for programs supported by their tool. As expected, the inferred bound grows as we increase the size parameter $n$---larger versions of the programs perform more operations, and so incur more backward error.

We find that our tool performs well on programs with few variables and operations, but performance slows as the number of variables in the program increases. We conjecture that this behavior is due to the larger space that must be searched by $\egglog$. Indeed, an analysis of the rules applied by $\egglog$ shows that associativity and commutativity rewrites for variable contexts are significant.
Memory usage is also a significant constraint; we find that programs larger than the ones evaluated run out of memory. Future work may investigate more efficient strategies for searching the space of structural lenses. 

Second, in \zcref{tab:second}, we created several challenging examples of backward stable programs that reuse variables. The (by-hand) backward error analyses require complex manipulation of error terms, none of which can be analyzed by existing tools. The full programs are given in \ifarxiv\zcref{app:bench}.\else\cite{supplement}.\fi~ We found that our implementation was very efficient for these programs and saturated within a few tenths of a second. 

\input{figures/benchmarks.tex}

%% file: figures/relation.tex
\begin{mathpar}
    \inferrule*[right=(Trans)]
    {\RR(\Gamma_1,E_1,\Gamma_2) \\ \RR(\Gamma_2,E_2,\Gamma_3) \\\\ 
        E_3=E_2[E_1/\Var(\Gamma_2)]}
    {\RR(\Gamma_1,E_3,\Gamma_3)}

    \inferrule*[right=(Cong)]
    {\RR(\Gamma_2,E,\Gamma_3)}
    {\RR(\Gamma_1\otimes \Gamma_2, (\id_{\Gamma_1}, E),\Gamma_1\otimes\Gamma_3)}

    \inferrule*[right=($\pi_2$)]
    { }
    {\RR(B_1\star B_2,\id_{B_2},B_2)}

    \inferrule*[right=(Assoc)]
    {\Gamma=B_1\otimes\dots\otimes B_m}
    {\RR((B\star\Gamma)\otimes B_{m+1}, (\id_{B};\id_{\Gamma},\id_{B_{m+1}}),B\star(\Gamma\otimes B_{m+1}))}

    \inferrule*[right=(Share)]
    {B_1=(\vec{x}:(0,\R^n,p)) \\ B_2=(\vec{y}:(0,\R^m,p)) \\\\ 
        B_3=(\vec{z}:(0,\R^{n+m},p))}
    {\RR(B_1\otimes B_2,(\id_{B_1},\id_{B_2}),B_3)}

    \inferrule*[right=($\Delta$)]
    {B_1=(\vec{x}:(0, \R^n,p)) \\\\ 
        B_2=(\vec{y}:(0,\R^{2n},p))}
    {\RR(B_1,(\id_{B_1},\id_{B_1}),B_2)}

    \inferrule*[right=(Add-$\star$)]
    {B_1=(x_1,x_2:(i,\R^2,p+1)) \\ B_2=(y:(i, \R,p))}
    {\RR(B\star B_1,(\id_{B};\add~x_1~x_2),B\star B_2)}

    \inferrule*[right=(Add)]
    {B_1=(x_1,x_2:(0,\R^2,p+1)) \\\\
        B_2=(y:(0, \R,p))}
    {\RR(B_1,\add~x_1~x_2,B_2)}
    
    \inferrule*[right=(Mul)]
    {B_1=(x_1,x_2:(0,\R^2,(p+1)/2)) \\\\
        B_2=(y:(0, \R,p))}
    {\RR(B_1,\mul~x_1~x_2,B_2)}

    \inferrule*[right=(DMul)]
    {B_1=(x_1:(0,\R,p)) \\\\ 
        B_2=(x_2:(0,\R,q+1)) \\ B_3=(y:(1,\R,q))}
    {\RR(B_1\otimes B_2,(x_1;\mul~x_1~x_2),B_1\star B_3)}

    \inferrule*[right=(Sqrt)]
    {B_1=(x:(0,\R,2p+2)) \\\\ 
        B_2=(y:(0,\R,p))}
    {\RR(B_1,\sqr~x,B_2)}
\end{mathpar}

%% file: figures/benchmarks.tex
\begin{table}[htbp]
\small\centering
\caption{Running $\eggshel$ on variable-size programs. \textbf{Vars} gives the number of variables, \textbf{Ops} gives the number of floating-point operations, \textbf{Err} gives the maximum per-variable backward error bound in terms of $\e$ (lower is better). When multiple sets of bounds are returned, we report the smallest maximum bound. \textbf{Time} gives the $\eggshel$ timing in seconds, and \textbf{$\Bean$} gives the $\Bean$ timing in seconds, when applicable.}
\begin{tabular}{lccccc @{\hskip 2em} lccccc}
    \toprule
    \textbf{Name} & \textbf{Vars} & \textbf{Ops} & \textbf{Err} & \textbf{Time} & \textbf{$\Bean$} 
    & \textbf{Name} & \textbf{Vars} & \textbf{Ops} & \textbf{Err} & \textbf{Time} & \textbf{$\Bean$} \\
    \midrule
    \multirow{6}{*}{\texttt{sum}} 
    & 5 & 4 & 4 & 0.04 & 0.0009 & 
    \multirow{6}{*}{\texttt{linear}} 
    & 2 & 2 & 1 & 0.03 & n/a \\ \cmidrule(lr{2em}){2-6} \cmidrule(lr){8-12}
    & 10 & 9 & 9 & 1.83 & 0.001 &
    & 3 & 4 & 2 & 0.03 & n/a \\ \cmidrule(lr{2em}){2-6} \cmidrule(lr){8-12}
    & 11 & 10 & 10 & 6.49 & 0.001 & 
    & 4 & 6 & 3 & 0.11 & n/a \\ \cmidrule(lr{2em}){2-6} \cmidrule(lr){8-12}
    & 12 & 11 & 11 & 21.81 & 0.001 & 
    & 5 & 8 & 4 & 1.56 & n/a \\ \cmidrule(lr{2em}){2-6} \cmidrule(lr){8-12}
    & 13 & 12 & 12 & 69.44 & 0.001 &
    & 6 & 10 & 5 & 23.22 & n/a \\ \cmidrule(lr{2em}){2-6} \cmidrule(lr){8-12}
    & 14 & 13 & 13 & 379.88 & 0.001 &
    & 7 & 12 & 6 & 600.16 & n/a \\ 
    \midrule
    \multirow{7}{*}{\texttt{norm}} 
    & 1 & 2 & 1.5 & 0.03 & n/a &
    \multirow{4}{*}{\texttt{quad}} 
    & 2 & 3 & 3 & 0.03 & n/a \\ \cmidrule(lr{2em}){2-6} \cmidrule(lr){8-12}
    & 2 & 4 & 2 & 0.02 & n/a &
    & 3 & 6 & 4 & 0.07 & n/a \\ \cmidrule(lr{2em}){2-6} \cmidrule(lr){8-12}
    & 3 & 6 & 2.5 & 0.04 & n/a & 
    & 4 & 9 & 5 & 3.50 & n/a \\ \cmidrule(lr{2em}){2-6} \cmidrule(lr){8-12}
    & 4 & 8 & 3.0 & 0.12 & n/a & 
    & 5 & 12 & 6 & 243.41 & n/a \\ \cmidrule(lr){2-12} 
    & 5 & 10 & 3.5 & 1.36 & n/a & 
    \multirow{3}{*}{\texttt{dotprod}} 
    & 4 & 3 & 1 & 0.04 & 0.0009 \\ \cmidrule(lr{2em}){2-6} \cmidrule(lr){8-12}
    & 6 & 12 & 4 & 16.96 & n/a & 
    & 6 & 5 & 2 & 0.80 & 0.0009 \\ \cmidrule(lr{2em}){2-6} \cmidrule(lr){8-12}
    & 7 & 14 & 4.5 & 298.62 & n/a &
    & 8 & 7 & 2 & 127.82 & 0.001 \\
    \bottomrule
\end{tabular}\label{tab:first}
\end{table}

\begin{table}[htbp]
\small\centering
\caption{Running $\eggshel$ on programs with more complex backward error analyses. The exact functions are given in the \textbf{Program} column, with full programs given in \ifarxiv\zcref{app:bench}.\else\cite{supplement}.\fi}
\begin{tabular}{lcccc @{\hskip 2em} lcccc}
    \toprule
    \textbf{Program} & \textbf{Vars} & \textbf{Ops} & \textbf{Err} & \textbf{Time} & \textbf{Program} & \textbf{Vars} & \textbf{Ops} & \textbf{Err} & \textbf{Time} \\
    \midrule 
    $x+(ax+bx^2)$ & 3 & 5 & 4 & 0.04 & 
    $a+a\sqrt{b}$ & 2 & 3 & 4 & 0.03 \\
    \midrule
    $a+\sqrt{ab}$ & 2 & 3 & 4 & 0.04 &
    $\sqrt{ax+\sqrt{b}}$ & 3 & 4 & 8 & 0.04 \\
    \midrule
    $(a+b)^2$ & 2 & 3 & 1.5 & 0.02 &
    $(a + \sqrt{b})^2$ & 2 & 5 & 5 & 0.04 \\
    \midrule
    $(a+ab)(c+cd)$ & 4 & 5 & 1.5 & 0.27 &
    $\sqrt{a}\sqrt{b}$ & 2 & 3 & 3 & 0.02 \\
    \bottomrule
\end{tabular}\label{tab:second}
\end{table}

%% file: sections/related.tex
\paragraph{Automated Backward Error Analysis} Existing tools for estimating the backward error of a program use automatic differentiation, optimization, and type system-based approaches. Miller and Spooner's Algorithm 532 \cite{miller} computes a program's partial derivatives and uses them to estimate backward error. \citet{rowan}, in a technique called functional stability analysis, uses estimates on forward error and function sensitivity to calculate the backward error. \citet{fu} employ optimization to find the smallest input perturbation to the exact function that yields the computed output. Both \citet{rowan} and \citet{fu} treat a higher-precision version of the program as the exact function. These techniques differ from ours as they compute approximate bounds, often running the programs directly. \citet{bean} offered the first automated method to give sound results. They define a type system whose semantics ensure that any typeable program is backward stable. Similar to our work, they base their analysis on a category of lenses.

\paragraph{Floating-Point Verification} Sound tools for bounding floating-point error generally employ abstract interpretation, interval arithmetic, and theorem proving. Almost all the tools we surveyed compute forward error bounds rather than backward, and hence use fundamentally different techniques from ours. FLUCTUAT \cite{fluctuat}, PRECiSA \cite{precisa}, and Gappa \cite{gappa} use interval arithmetic combined with abstract interpretation or formal proofs to statically bound error at each floating-point operation. In contrast, FPTaylor \cite{fptaylor} optimizes over a symbolic Taylor series expansion of a given program and generates certificates of soundness. VCFloat \cite{vcfloat} is a Rocq library for writing formal proofs about error analysis. NumFuzz \cite{numfuzz}, closely related to $\Bean$ \cite{bean}, is a type system bounding forward error.

\paragraph{Backward Error in Numerical Analysis} Backward error has been a key concept in numerical computation for many decades, first popularized by \citet{wilkinson}. Other important works on this topic include \citet{higham_linear}, which introduced structured backward error for linear systems.  The textbooks by \citet{higham} and \citet{corless} provide many further examples of backward error analysis in numerical analysis.

\paragraph{Lenses and Lens Categories} Backward error lenses are inspired by lenses from bidirectional programming \cite{lenses}. The category $\Shel$ shares many similarities with other categories of lenses (e.g., \cite{depaiva}); for instance, it is symmetric monoidal but not Cartesian. The idea of using a lens to track shifts or perturbations is also broadly similar to edit lenses~\citep{editlens}, where the backward flow of information tracks an edit (or a ``delta'').

\paragraph{{\normalfont\texttt{egglog}}-Based Synthesizers} There are many tools leveraging the \texttt{egg} and \texttt{egglog} solvers. The closest to our work is probably Herbie \cite{herbie}, a tool for improving floating-point expressions by sampling error and applying rewrites. Herbie can be implemented in $\egglog$ \cite{egglog} to ensure all its rewrites are sound. Herbie differs from our analysis in that it attempts to improve a floating-point implementation using heuristics and sampling, while we perform a purely static analysis and use $\egglog$ to search for proofs of stability.

%% file: sections/conclusion.tex
We present a framework for backward error analysis for programs that were previously not analyzable. We offer a new definition of a backward error lens and develop the category $\Shel$, which allows us to construct these lenses. Using this category as our semantics, we implement the tool $\eggshel$, which searches for lenses that prove a given program backward stable. We evaluate our tool on known backward stable floating-point programs. We see several promising directions for future work.

\paragraph{Optimizing the Implementation} By adding rules and operations to our relation in \zcref{sec:implementation}, we would be able to support automated analysis for more programs. However, adding rules may cause the analysis to become slower. Investigating optimizations as well as more efficient methods for synthesizing backward error lenses would be valuable.

\paragraph{Conditional Backward Stability} Many proofs of backward stability in the numerical analysis literature rely on conditions on the input. For example, to return the full matrix $L$ in a Cholesky factorization in a backward stable manner requires a rewrite based on the symmetric positive definite condition. It may be possible to encode these conditions into our category, which would allow us to perform these types of analyses.

\paragraph{Probabilistic Backward Error} Finally, recent works in numerical analysis explore backward error bounds under stochastic models of rounding supported by emerging hardware. For example, \citet{probbe} demonstrate that common numerical algorithms can have significantly less backward error under probabilistic rounding. It would be interesting to extend our lens category to handle such programs.

%% file: sections/appendix.tex
\section{Backward Error Analysis By Hand}\label{app:analysis_by_hand}
\input{sections/appendix/analysis_by_hand}

\section{Proofs From Section 3: Relational Backward Error Lenses}\label{app:lenses_proofs}
\input{sections/appendix/lenses_proofs}

\section{Supplementary Material for the Category Shel}\label{app:category_proofs}
\input{sections/appendix/category_proofs}

\section{Benchmark Programs}\label{app:bench}
\input{sections/appendix/bench}

%% file: sections/appendix/analysis_by_hand.tex
We will now show an example of performing backward error analysis by hand. This is the type of analysis that our automation emulates. It is a well-known fact that the following dot product program is backward stable. Let $\tf:\F^2\times\F^2\to \F$ be the program
\begin{equation*}
    \tf((x_1,x_2),(y_1,y_2))=\add~(\mul~x_1~y_1)~(\mul~x_2~y_2).
\end{equation*}
This program models the function $f:\R^2\times\R^2\to\R$ which takes the dot product of two input vectors. Let us prove standard backward stability (\zcref{def:std_def}) for inputs $(x_1,x_2),(y_1,y_2)\in\F^2$:
\begin{align*}
    \tf((x_1,x_2),(y_1,y_2)) &= \add~(\mul~x_1~y_1)~(\mul~x_2~y_2) \\
    &= \add~(x_1y_1e^{\delta_1})~(x_2y_2e^{\delta_2}) &\tag*{by \zcref{def:olver_model}} \\
    &= (x_1y_1e^{\delta_1}+x_2y_2e^{\delta_2})e^{\delta_3} &\tag*{by \zcref{def:olver_model}} \\
    &= x_1y_1e^{\delta_1+\delta_3}+x_2y_2e^{\delta_2+\delta_3} \\
    &= \tx_1\ty_1+\tx_2\ty_2 
\end{align*}
where
\begin{align*}
    \tx_1=x_1e^{(\delta_1+\delta_3)/2} \quad \tx_2=x_2e^{(\delta_2+\delta_3)/2} \quad \ty_1=y_1e^{(\delta_1+\delta_3)/2} \quad \ty_2=y_2e^{(\delta_2+\delta_3)/2}
\end{align*}
and $|\delta_1|,|\delta_2|,|\delta_3|\leq \e$. Thus, 
\begin{equation*}
    d(x_1,\tx_1)=d(x_1,x_1e^{(\delta_1+\delta_3)/2})=|\delta_1+\delta_3|/2\leq \e,
\end{equation*}
and similarly for each of the other variables. We conclude that $\tf$ is a $(1,1,1,1)$-backward stable program modeling dot product. To get a sense for the compositional nature of the analysis, see that the inner $\mul$ programs push $\e/2$ backward error onto each of the input variables. The larger $\add$ program pushes an additional $\e$ error onto each of its operands, which are the $\mul$ subprograms, and that error is distributed evenly onto the variables.

%% file: sections/appendix/lenses_proofs.tex
This appendix contains proofs deferred from \zcref{sec:lenses}.

\lensComp*
\begin{proof}
    Let $(f_1,\tf_1,b_1)$ be an $(\alpha,\beta)$-relational lens and $(f_2,\tf_2,b_2)$ be a $(\beta,\gamma)$-relational lens both satisfying \zcref{def:shel_def}. We form the $(\alpha,\gamma)$-lens by
    \begin{equation*}
        (f_2,\tf_2,b_2)\circ (f_1,\tf_1,b_1)=(f_2\circ f_1,\tf_2\circ\tf_1,b)
    \end{equation*}
    where 
    \begin{equation*}
        b(x,z)=b_1(x, b_2(\tf_1(x),z)).
    \end{equation*}
    Now, take an input $x$ in the domain of $\tf_1$ and a value $z$ in the range of $f_2$ such that 
    \begin{equation*}
        d(\tf_2(\tf_1(x)),z)\leq \gamma\e.
    \end{equation*}
    By the distance condition on $\tf_2$,
    \begin{equation*}
        d(\tf_1(x),b_2(\tf_1(x),z)) \leq \beta\e
    \end{equation*}
    so by the distance condition on $\tf_1$,
    \begin{equation*}
        d(x,b_1(x,b_2(\tf_1(x),z)))=d(x,b(x,z))\leq \alpha\e.
    \end{equation*}
    Finally, we show that 
    \begin{align*}
        (f_2\circ f_1)(\tx) &= f_2(f_1(b_1(x,b_2(\tf_1(x),z)))) \\
        &= f_2(b_2(\tf_1(x),z)) \tag*{by the exactness condition on \text{$\tf_1$}} \\
        &= z. \tag*{by the exactness condition on \text{$\tf_2$}}
    \end{align*}
\end{proof}

\sqrtLens*
\begin{proof}
        Given $x\in\F$ and $y\in \R$ such that $d(\tf(x),y)\leq p\e$, we know from the definition of the $RP$ metric (\zcref{def:rp_metric}) that $y\geq 0$. Expanding the bound: 
    \begin{equation*}
        d(\tf(x),y)=\left|\frac{1}{2}\ln|x|+\delta-\ln (y)\right|\leq p\e
    \end{equation*}
    and therefore, by the triangle inequality,
    \begin{equation*}
        \frac{1}{2}\left|\ln|x|-2\ln (y)\right|\leq p\e+\e.
    \end{equation*}
    We compute that
    \begin{equation*}
        f(\tx)=\sqrt{y^2}=y
    \end{equation*}
    and
    \begin{equation*}
        d(x,\tx)=|\ln |x|-2\ln (y)|\leq (2p+2)\e. \qedhere
    \end{equation*} 
\end{proof}

\logLens*
\begin{proof}
    Given $x\in[1,a]\cap\F$ and $y\in\R$ such that $d(\tf(x),y)\leq p\e$, the bound tells us 
    \begin{equation*}
        |\ln(\ln(x))+\delta-\ln (y)|\leq p\e \quad\text{and therefore}\quad  |\ln(\ln(x))-\ln (y)|\leq (p+1)\e.
    \end{equation*}
    From this we also get a bound
    \begin{equation*}
        y\leq e^{(p+1)\e} \ln(x)\leq e^{(p+1)\e}\ln(a)\leq e\ln(a)\leq 3a,
    \end{equation*}
    reasonably assuming $(p+1)\e\leq 1$. We need to prove that 
    \begin{equation*}
        d(x,\tx) = |\ln(x) - y|\leq 3a(p+1)\e.
    \end{equation*}
    By the mean value theorem, there is some $c$ between $y$ and $\ln(x)$ such that 
    \begin{equation*}
        \frac{1}{c}=\frac{\ln(\ln(x))-\ln(y)}{\ln(x)-y}
    \end{equation*}
    so
    \begin{equation*}
        |\ln(x)-y| \leq |c|\cdot |\ln(\ln(x))-\ln(y)| \leq \max\{\ln(x),y\}\cdot(p+1)\e \leq 3a(p+1)\e. \qedhere
    \end{equation*}
\end{proof}

%% file: sections/appendix/category_proofs.tex
This appendix contains proofs deferred from \zcref{sec:category} and \zcref{sec:case_studies} as well as additional results about the category $\Shel$. 

\morphismComp*
\begin{proof}
    We check the conditions for composition. For $x\in \dom(\tf_1)$ and $u\in U$ such that $\|u\|\leq r$, it follows that $\|b_2(u)\|\leq q$ so we can use the conditions for $f_1$. Now,
    \begin{align*}
        (1) \quad f(x* b(u))&=f_2(f_1(x* b_1(b_2(u)))) & (2) \quad \|b(u)\|&=\|b_1(b_2(u))\| \\
        &= f_2(\tf_1(x)* b_2(u)) =\tf_2(\tf_1(x))* u & &\leq p. \\
        &= \tf(x)* u 
    \end{align*}
    Associativity of composition is straightforward.
\end{proof}

\symMon*
We check a few of the requirements for showing that $\Shel$ is a symmetric monoidal category. First, we prove that given $f_1:X_1\to Y_1$ and $f_2:X_2\to Y_2$, their parallel composition is a morphism 
\begin{equation*}
    f_1\otimes f_2:X_1\otimes X_2\to Y_1\otimes Y_2 \quad\text{with maps}\quad ((f_1, f_2), (\tf_1,\tf_2), (b_1,b_2)).
\end{equation*}
\begin{proof}
    Take $(x_1,x_2)\in \dom(\tf_1\times\tf_2)$ and $(t_1,t_2)\in T_1\times T_2$ such that $\|t_1\|\leq q_1$ and $\|t_2\|\leq q_2$. 
    \begin{align*}
        (1) \quad (f_1,f_2)&((x_1,x_2)*b(t_1,t_2)) & (2) \quad \|b(t_1,t_2)\| &=\|(b_1(t_1),b_2(t_2))\| \\
        &= (f_1,f_2)(x_1*b_1(t_1),x_2*b_2(t_2)) && =(\|b_1(t_1)\|,\|b_2(t_2)\|) \\
        &= (f_1(x_1*b_1(t_1)), f_2(x_2*b_2(t_2))) && \leq (p_1,p_2). \\
        &= (\tf_1(x_1)*t_1,\tf_2(x_2)*t_2) \\
        &= (\tf_1,\tf_2)(x_1,x_2)*(t_1,t_2)
    \end{align*}
    Thus, parallel composition is a valid morphism.
\end{proof}

Next, we show we have a left unitor, a morphism $\lambda:I\otimes X\to X$. The right unitor is defined similarly.
\begin{proof}
    Let 
    \begin{equation*}
        f(\diamond,x)=\tf(\diamond,x)=x\qquad b(s)=(e,s).
    \end{equation*}
    For $x\in X$ and $\|s\|\leq p$,
    \begin{equation*}
        (1) \quad f((\diamond,x)*b(s)) = f(\diamond,x*s) = x*s = \tf(\diamond,x)*s \qquad
        (2) \quad \|b(s)\| =\|(e,s)\| \leq (0,p),
    \end{equation*}
    satisfying the conditions.
\end{proof}

Finally, we show we have a swap map, a morphism $s:X_1\otimes X_2\to X_2\otimes X_1$.
\begin{proof}
    Let 
    \begin{equation*}
        f(x_1,x_2)=\tf(x_1,x_2)=(x_2,x_1) \qquad b(s_2,s_1)=(s_1,s_2).
    \end{equation*}
    For $(x_1,x_2)\in X_1\times X_2$ and $\|s_2\|\leq p_2$, $\|s_1\|\leq p_1$,
    \begin{align*}
        (1) \quad f&((x_1,x_2)*b(s_2,s_1)) = f(x_1*s_1,x_2*s_2) & (2) \quad \|b(s_2,s_1)\| &\leq (p_1,p_2) \\
        &= (x_2*s_2,x_1*s_1) =(x_2,x_1)*(s_2,s_1) \\
        &=\tf(x_1,x_2)*(s_2,s_1) 
    \end{align*}
    satisfying the conditions.
\end{proof}

\addSubMul*
\begin{proof}
    For $\add$, let 
    \begin{equation*}
        f(x,y)=x+y \qquad \tf(x,y)=(x+y)e^\delta \qquad b(s)=s+\delta.
    \end{equation*}
    For $x,y\in\F$ and $|s|\leq p$,
    \begin{align*}
        (1)\quad f((x,y)*(s+\delta)) &= f(xe^{s+\delta},ye^{s+\delta})=xe^{s+\delta}+ye^{s+\delta} & (2)\quad |s+\delta|&\leq p+\e. \\
        &= (x+y)e^{\delta}e^s = \tf(x,y)*s
    \end{align*}

    For $\mul$, let
    \begin{equation*}
        f(x,y)=xy \qquad \tf(x,y)=xye^\delta \qquad b(s)=(s+\delta)/2.
    \end{equation*}
    For $x,y\in\F$ and $|s|\leq p$,
    \begin{align*}
        (1)\quad f((x,y)*(s+\delta)/2) &= f(xe^{(s+\delta)/2},ye^{(s+\delta)/2}) &(2)\quad |(s+\delta)/2| &\leq (p+\e)/2, \\
        &= xye^{s+\delta} =\tf(x,y)*s
    \end{align*}
    satisfying the conditions.
\end{proof}

\pushAdjust*
\begin{proof}
    Let
    \begin{equation*}
        f=\tf=\id \qquad b(s_1,s_2)=(s_1,i(s_1)^{-1}*j(s_1)*s_2). 
    \end{equation*}
    For $(x_1,x_2)\in X_1\times X_2$ and $\|s_1\|\leq p_1$, $\|s_2\|\leq p_2$, 
    \begin{align*}
        (1)\quad f&((x_1,x_2)*(s_1,i(s_1)^{-1}*j(s_1)*s_2)) & (2) \quad \|&(s_1,i(s_1)^{-1}*j(s_1)*s_2)\| \\
        &=f(x_1*s_1,x_2*i(s_1)*i(s_1)^{-1}*j(s_1)*s_2) &&= (\|s_1\|,\|i(s_1)^{-1}*j(s_1)*s_2\|) \\
        &= f(x_1*s_1,x_2*j(s_1)*s_2) &&\leq (\|s_1\|,\|i(s_1)^{-1}*j(s_1)\|+\|s_2\|) \\
        &= \tf(x_1,x_2)*(s_1,s_2) &&\leq (p_1, p+p_2),
    \end{align*}
    satisfying the conditions.
\end{proof}

\projections*
\begin{proof}
    For $\pi_1$, let
    \begin{equation*}
        f(x_1,x_2)=\tf(x_1,x_2)=x_1 \qquad b(s_1)=(s_1,e_2).
    \end{equation*}
    For $(x_1,x_2)\in X_1\times X_2$ and $\|s_1\|\leq p_1$,
    \begin{align*}
        (1)\quad f((x_1,x_2)*(s_1,e_2)) &= f(x_1*s_1,x_2*i(s_1)*e_2) & (2) \quad \|(s_1,e_2)\| &\leq (p_1,0) \\
        &= x_1*s_1 && \leq (p_1,p_2). 
    \end{align*}
    For $\pi_2$, let 
    \begin{equation*}
        f(x_1,x_2)=\tf(x_1,x_2)=x_2 \qquad b(s_2)=(e_1,s_2).
    \end{equation*}
    For $(x_1,x_2)\in X_1\times X_2$ and $\|s_2\|\leq p_2$,
    \begin{align*}
        (1)\quad f((x_1,x_2)*(e_1,s_2)) &= f(x_1*e_1,x_2*i(e_1)*s_2) & (2) \quad \|(e_1,s_2)\| &\leq (p_1,p_2), \\
        &= f(x_1,x_2*s_2) =x_2*s_2
    \end{align*}
    satisfying the conditions.
\end{proof}

\distributor*
\begin{proof}
    Let 
    \begin{align*}
        f((x_1,y_1),(x_2,y_2))&=\tf((x_1,y_1),(x_2,y_2))=((x_1,x_2),(y_1,y_2)) \\
        b((s_1,s_2),(t_1,t_2))&=((s_1,t_1),(s_2,t_2)). 
    \end{align*}
    For $((x_1,y_1),(x_2,y_2))\in (X_1\times Y_1)\times (X_2\times Y_2)$ and $\|s_1\|\leq p_1$, $\|s_2\|\leq p_2$, $\|t_1\|\leq q_1$, $\|t_2\|\leq q_2$,
    \begin{align*}
        (1)\quad f&(((x_1,y_1),(x_2,y_2))*((s_1,t_1),(s_2,t_2))) & (2) \quad \|&((s_1,t_1),(s_2,t_2))\|\\
        &= f((x_1,y_1)*(s_1,t_1),(x_2,y_2)*(s_2,t_2)) &&\leq ((p_1,q_1),(p_2,q_2)), \\
        &= f((x_1*s_1,y_1*i(s_1)*t_1),(x_2*s_2,y_2*j(s_2)*t_2)) \\
        &= ((x_1*s_1,x_2*s_2),(y_1*i(s_1)*t_1,y_2*j(s_2)*t_2)) \\
        &= ((x_1,x_2)*(s_1,s_2),(y_1,y_2)*(i,j)(s_1,s_2)*(t_1,t_2)) \\
        &= \tf((x_1,y_1),(x_2,y_2))*((s_1,s_2),(t_1,t_2))
    \end{align*}
    satisfying the conditions.
\end{proof}

\begin{lemma}[Parallel Composition for Push Product]\label{lem:push_parallel}
    Suppose 
    \begin{align*}
        f_1:(X_1,S_1,p_1)\to (Y_1,T_1,q_1) \qquad f_2:(X_2,S_2,p_2)\to (Y_2,T_2,q_2)
    \end{align*}
    and $i:S_1\to S_2$, $j:T_1\to T_2$ are homomorphisms. Further suppose that for all $(x_1,x_2)\in X_1\times X_2$, $(t_1,t_2)\in T_1\times T_2$ satisfying $\|(t_1,t_2)\|\leq (q_1,q_2)$, 
    \begin{equation*}
        f_2(x_2* b_2(t_2)* i(b_1(t_1)))=f_2(x_2* b_2(t_2))* j(t_1).
    \end{equation*}
    Then we can define a morphism $f_1\star f_2:X_1\star_i X_2\to Y_1\star_j Y_2$ where $f_1$ and $f_2$ are applied in parallel.
\end{lemma}
\begin{proof}
    For $(x_1,x_2)\in \dom(\tf_1\times\tf_2)$ and $\|t_1\|\leq q_1$, $\|t_2\|\leq q_2$,
    \begin{align*}
        (1) \quad f&((x_1,x_2)* b(t_1,t_2))
        & (2) \quad \|&b(t_1,t_2)\|\\
        &= f((x_1,x_2)* (b_1(t_1),b_2(t_2))) &&=\|(b_1(t_1),b_2(t_2))\| \\
        &= f(x_1* b_1(t_1), x_2*i(b_1(t_1))*b_2(t_2)) &&= (\|b_1(t_1)\|,\|b_2(t_2)\|) \\
        &= (f_1(x_1*b_1(t_1)),f_2(x_2* i(b_1(t_1))* b_2(t_2))) &&\leq (p_1,p_2), \\
        &= (\tf_1(x_1)* t_1, f_2(x_2* b_2(t_2))* j(t_1))  \\
        &= (\tf_1(x_1)* t_1,\tf_2(x_2)* t_2* j(t_1)) \\
        &= \tf(x_1,x_2)* (t_1,t_2)
    \end{align*}
    satisfying the conditions.
\end{proof}

\shareLens*
\begin{proof}
    Let 
    \begin{equation*}
        f=\tf=\id \qquad b(s)=(s,i(s)^{-1}*s).
    \end{equation*}
    For $(x_1,x_2)\in X_1\times X_2$ and $\|s\|\leq p_1$,
    \begin{align*}
        (1)\quad f&((x_1,x_2)*(s,i(s)^{-1}*s)) &(2) \quad \|(s,i(s)^{-1}*s)\| &\leq (p_1,p_2), \\
        &= f(x_1*s,x_2*i(s)*i(s)^{-1}*s) \\
        &= f(x_1*s,x_2*s) = \tf(x_1,x_2)*s
    \end{align*}
    satisfying the conditions.
\end{proof}

\begin{lemma}\label{lem:div_lens}
    We have a division lens
    \begin{equation*}
        (\R,\R^+,p+\e)\otimes (\R,\R^+,0)\to (\R,\R^+,p).
    \end{equation*}
\end{lemma}
\begin{proof}
    Define 
    \begin{align*}
        f(x,y)=0\text{ if }y=0\text{ else }x/y \qquad 
        \tf(x,y)=0 \text{ if }y=0\text{ else }xe^\delta/y \qquad
        b(s)&=(s+\delta,0).
    \end{align*}
    For $x,y\in\F$ and $|s|\leq p$, if $y=0$, then 
    \begin{align*}
        f((x,y)*(s+\delta,0)) &= 0=\tf(x,y)*s.
    \end{align*}
    Otherwise,
    \begin{align*}
        (1) \quad f&((x,y)*(s+\delta,0)) = f(xe^{s+\delta},y) &(2) \quad (|s+\delta|,0) &\leq (p+\e,0), \\
        &= (xe^{\delta}/y)e^s =\tf(x,y)*s
    \end{align*}
    satisfying the conditions.
\end{proof}

\begin{lemma}\label{lem:add_div_lens}
    We have an addition/division lens
    \begin{equation*}
        (\R\times\R,\R^+,p_1)\star_i(\R,\R^+,p_2+2\e)\to (\R,\R^+,p_2).
    \end{equation*}
\end{lemma}
\begin{proof}
    Define
    \begin{gather*}
        f(x,y,z)=0\text{ if }x+y=0\text{ else }z/(x+y) \qquad
        \tf(x,y,z)=0 \text{ if }x+y=0\text{ else }ze^{\delta_2-\delta_1}/(x+y) \\ 
        b(s) = (0,s+\delta_2-\delta_1).
    \end{gather*}
    For $x,y,z\in\F$ and $|s|\leq p_2$, if $x+y=0$, then 
    \begin{align*}
        f((x,y,z)*(0,s+\delta_2-\delta_1)) &= 0=\tf(x,y,z)*s.
    \end{align*}
    Otherwise,
    \begin{align*}
        (1) \quad f&((x,y,z)*(0,s+\delta_2-\delta_1)) = f(x,y,ze^{s+\delta_2-\delta_1}) &(2) \quad (0,|s+\delta_2-\delta_1|) &\leq (p_1,p_2+2\e), \\
        &= (ze^{\delta_2-\delta_1}/(x+y))e^s =\tf(x,y,z)*s
    \end{align*}
    satisfying the conditions.
\end{proof}

%% file: sections/appendix/bench.tex
Here, we detail the benchmark programs evaluated in \zcref{tab:second}. For each example, we show one sample backward error bound synthesized by our tool.
\begin{enumerate}
    \item This program is 
    \begin{lstlisting}
        (Add x (Add (Mul a x) (Mul (Mul b x) x)))
    \end{lstlisting}
    implementing the function 
    \begin{equation*}
        f(a, b,x)=x+ax+bx^2.
    \end{equation*}
    The backward error bound on $a$ is $2\e$, on $b$ is $4\e$, and on $x$ is $\e$.
    \item This program is 
    \begin{lstlisting}
        (Add a (Sqrt (Mul a b)))
    \end{lstlisting}
    implementing the function 
    \begin{equation*}
        f(a, b)=a+\sqrt{ab}.
    \end{equation*}
    The backward error bound on $a$ is $\e$ and on $b$ is $4\e$.
    \item This program is 
    \begin{lstlisting}
        (Mul (Add a b) (Add b a))
    \end{lstlisting}
    implementing the function
    \begin{equation*}
        f(a, b)=(a+b)^2.
    \end{equation*}
    The backward error bound on $a$ is $3\e/2$ and on $b$ is $3\e/2$.
    \item This program is 
    \begin{lstlisting}
        (Mul (Add a (Mul a b)) (Add c (Mul c d)))
    \end{lstlisting}
    implementing the function
    \begin{equation*}
        f(a, b,c,d)=(a+ab)(c+cd).
    \end{equation*}
    The backward error bound on $a$ is $3\e/2$, on $b$ is $\e$, on $c$ is $3\e/2$, and on $d$ is $\e$.
    \item This program is 
    \begin{lstlisting}
        (Add a (Mul a (Sqrt b)))
    \end{lstlisting}
    implementing the function
    \begin{equation*}
        f(a,b)=a+a\sqrt{b}.
    \end{equation*}
    The backward error bound on $a$ is $\e$ and on $b$ is $4\e$.
    \item This program is 
    \begin{lstlisting}
        (Sqrt (Add (Mul a x) (Sqrt b)))
    \end{lstlisting}
    implementing the function
    \begin{equation*}
        f(a,b,x)=\sqrt{ax+\sqrt{b}}.
    \end{equation*}
    The backward error bound on $a$ is $0$, on $b$ is $8\e$, and on $x$ is $4\e.$
    \item This program is 
    \begin{lstlisting}
        (Mul (Add a (Sqrt b)) (Add a (Sqrt b)))
    \end{lstlisting}
    implementing the function
    \begin{equation*}
        f(a,b)=(a + \sqrt{b})^2.
    \end{equation*}
    The backward error bound on $a$ is $3\e/2$ and on $b$ is $5\e$.
    \item This program is 
    \begin{lstlisting}
        (Mul (Sqrt a) (Sqrt b))
    \end{lstlisting}
    implementing the function
    \begin{equation*}
        f(a,b)=\sqrt{a}\sqrt{b}.
    \end{equation*}
    The backward error bound on $a$ is $3\e$ and on $b$ is $3\e$.
\end{enumerate}

%% file: header.bib
@STRING{springer = {Springer-Verlag} }

@STRING{eptcs    = "Electronic Proceedings in Theoretical Computer Science" }

@STRING{pacmpl   = "Proceedings of the {ACM} on Programming Languages" }

@STRING{toplas   = "ACM Transactions on Programming Languages and Systems" }

@STRING{ieee     = "IEEE" }

@STRING{popl12      = popl # ", Philadelphia, Pennsylvania" }

@STRING{cpp24       = cpp # ", London, England" }

@STRING{oopsla15    = oopsla # ", Pittsburgh, Pennsylvania" }

@STRING{pldi        = "{ACM SIGPLAN Conference on Programming Language Design
                      and Implementation (PLDI)}" }

@STRING{pldi15      = pldi # ", Portland, Oregon" }

@STRING{fm24        = fm # ", Milan, Italy" }

@STRING{sas06       = sas # ", Seoul, Korea" }

@preamble{"\newcommand{\SortNoop}[1]{}"}


%% file: references.bib
@article{bean,
    author = {Kellison, Ariel E. and Zielinski, Laura and Bindel, David and Hsu, Justin},
    title = {Bean: A Language for Backward Error Analysis},
    year = {2025},
    volume = {9},
    number = {PLDI},
    doi = {10.1145/3729324},
    journal = pacmpl,
    articleno = {221},
    numpages = {25},
    pages = {1838--1862},
}

@article{olver,
    author = {Olver, F. W. J.},
    title = {A New Approach to Error Arithmetic},
    journal = {SIAM Journal on Numerical Analysis},
    volume = {15},
    number = {2},
    pages = {368--393},
    year = {1978},
    doi = {10.1137/0715024},
}

@inproceedings{mech_olver,
    author = {Fan, Max and Kellison, Ariel E. and Pollard, Samuel D.},
    year = {2025},
    title = {Mechanizing {Olver's} Error Arithmetic},
    booktitle = {International Workshop on Verification of Scientific Software, Hamilton, Ontario},
    address = {Hamilton, Ontario, Canada},
    series = eptcs,
    volume = {432},
    publisher = {Open Publishing Association},
    pages = {37--47},
    doi = {10.4204/eptcs.432.6},
}

@article{egglog,
    author = {Zhang, Yihong and Wang, Yisu Remy and Flatt, Oliver and Cao, David and Zucker, Philip and Rosenthal, Eli and Tatlock, Zachary and Willsey, Max},
    title = {Better Together: Unifying {Datalog} and Equality Saturation},
    year = {2023},
    volume = {7},
    number = {PLDI},
    doi = {10.1145/3591239},
    journal = pacmpl,
    articleno = {125},
    numpages = {25},
    pages = {468--492},
}

@book{higham,
    author = {Higham, Nicholas J.},
    title = {Accuracy and Stability of Numerical Algorithms},
    publisher = {Society for Industrial and Applied Mathematics},
    address = {Philadelphia, PA, USA},
    year = {2002},
    doi = {10.1137/1.9780898718027},
    edition = {Second},
}

@book{corless,
    author = {Corless, Robert M. and Fillion, Nicolas},
    title = {A Graduate Introduction to Numerical Methods},
    publisher = springer,
    address = {New York, NY, USA}, 
    year = {2013},
    doi = {10.1007/978-1-4614-8453-0},
    edition = {First},
}

@article{numfuzz,
    author = {Kellison, Ariel E. and Hsu, Justin},
    title = {Numerical Fuzz: A Type System for Rounding Error Analysis},
    year = {2024},
    volume = {8},
    number = {PLDI},
    doi = {10.1145/3656456},
    journal = pacmpl,
    articleno = {226},
    numpages = {25},
    pages = {1954--1978},
}

@article{gappa,
    author = {de Dinechin, Florent and Lauter, Christoph and Melquiond, Guillaume},
    journal = {IEEE Transactions on Computers}, 
    title = {Certifying the Floating-Point Implementation of an Elementary Function Using {Gappa}}, 
    year = {2011},
    volume = {60},
    number = {2},
    pages = {242--253},
    doi = {10.1109/TC.2010.128},
}

@article{fptaylor,
    author = {Solovyev, Alexey and Baranowski, Marek S. and Briggs, Ian and Jacobsen, Charles and Rakamari\'{c}, Zvonimir and Gopalakrishnan, Ganesh},
    title = {Rigorous Estimation of Floating-Point Round-Off Errors with Symbolic Taylor Expansions},
    year = {2018},
    volume = {41},
    number = {1},
    issn = {0164-0925},
    doi = {10.1145/3230733},
    journal = toplas,
    articleno = {2},
    numpages = {39},
    pages = {1--39},
}

@inproceedings{precisa,
    author = {Titolo, Laura and Moscato, Mariano and Feliu, Marco A. and Masci, Paolo and Mu\~{n}oz, C\'{e}sar A.},
    title = {Rigorous Floating-Point Round-Off Error Analysis in {PRECiSA} 4.0},
    year = {2024},
    isbn = {978-3-031-71177-0},
    publisher = springer,
    address = {Cham, Switzerland},
    doi = {10.1007/978-3-031-71177-0_2},
    booktitle = fm24,
    pages = {20--38},
}

@inproceedings{fu,
    author = {Fu, Zhoulai and Bai, Zhaojun and Su, Zhendong},
    title = {Automated Backward Error Analysis for Numerical Code},
    year = {2015},
    isbn = {9781450336895},
    publisher = {Association for Computing Machinery},
    address = {New York, NY, USA},
    doi = {10.1145/2814270.2814317},
    booktitle = oopsla15,
    pages = {639--654},
}

@phdthesis{rowan,
    author = {Rowan, Thomas Harvey},
    school = {University of Texas at Austin},
    title = {Functional Stability Analysis of Numerical Algorithms},
    year = {1990},
    address = {USA},
}

@article{miller,
    author = {Miller, Webb and Spooner, David},
    title = {Algorithm 532: Software for Roundoff Analysis [Z]},
    year = {1978},
    publisher = {Association for Computing Machinery},
    address = {New York, NY, USA},
    volume = {4},
    number = {4},
    issn = {0098-3500},
    doi = {10.1145/356502.356497},
    journal = {ACM Transactions on Mathematical Software},
    pages = {388--390},
}

@inproceedings{fluctuat,
    author = {Goubault, Eric and Putot, Sylvie},
    title = {Static Analysis of Numerical Algorithms},
    year = {2006},
    isbn = {978-3-540-37758-0},
    publisher = springer,
    address = {Berlin, Heidelberg},
    doi = {10.1007/11823230_3},
    booktitle = sas06,
    pages = {18--34},
}

@inproceedings{vcfloat,
    author = {Appel, Andrew W. and Kellison, Ariel E.},
    title = {{VCFloat2}: Floating-Point Error Analysis in {Coq}},
    year = {2024},
    publisher = {Association for Computing Machinery},
    address = {New York, NY, USA},
    doi = {10.1145/3636501.3636953},
    booktitle = cpp24,
    pages = {14--29},
}

@book{wilkinson,
    author = {Wilkinson, James Hardy},
    title = {Rounding Errors in Algebraic Processes},
    publisher = {Society for Industrial and Applied Mathematics},
    year = {2023},
    doi = {10.1137/1.9781611977523},
    address = {Philadelphia, PA},
    note = {Originally published 1963},
}

@article{higham_linear,
    author = {Higham, Desmond J. and Higham, Nicholas J.},
    title = {Backward Error and Condition of Structured Linear Systems},
    journal = {SIAM Journal on Matrix Analysis and Applications},
    volume = {13},
    number = {1},
    pages = {162--175},
    year = {1992},
    doi = {10.1137/0613014},
}

@article{lenses,
    author = {Foster, J. Nathan and Greenwald, Michael B. and Moore, Jonathan T. and Pierce, Benjamin C. and Schmitt, Alan},
    title = {Combinators for Bidirectional Tree Transformations: A Linguistic Approach to the View-Update Problem},
    year = {2007},
    volume = {29},
    number = {3},
    issn = {0164-0925},
    doi = {10.1145/1232420.1232424},
    journal = toplas,
    pages = {17},
}

@phdthesis{depaiva,
    author = {de Paiva, Valeria Correa Vaz},
    year = {1989},
    title = {The Dialectica categories},
    school = {University of Cambridge, Computer Laboratory},
    doi = {10.48456/tr-213},
    number = {UCAM-CL-TR-213}
}

@inproceedings{herbie,
    author = {Panchekha, Pavel and Sanchez-Stern, Alex and Wilcox, James R. and Tatlock, Zachary},
    title = {Automatically Improving Accuracy for Floating Point Expressions},
    year = {2015},
    isbn = {9781450334686},
    publisher = {Association for Computing Machinery},
    address = {New York, NY, USA},
    doi = {10.1145/2737924.2737959},
    booktitle = pldi15,
    pages = {1--11},
}

@inproceedings{editlens,
    author = {Hofmann, Martin and Pierce, Benjamin and Wagner, Daniel},
    title = {Edit Lenses},
    year = {2012},
    isbn = {9781450310833},
    publisher = {Association for Computing Machinery},
    address = {New York, NY, USA},
    doi = {10.1145/2103656.2103715},
    booktitle = popl12,
    pages = {495--508},
}

@article{probbe,
    author = {Connolly, Michael P. and Higham, Nicholas J. and Mary, Theo},
    title = {Stochastic Rounding and Its Probabilistic Backward Error Analysis},
    journal = {SIAM Journal on Scientific Computing},
    volume = {43},
    number = {1},
    pages = {A566--A585},
    year = {2021},
    doi = {10.1137/20M1334796},
}

@software{eggshel,
    author = {Zielinski, Laura and Hsu, Justin},
    title = {eggshel: A Floating-Point Backward Error Analysis Tool},
    year = 2026,
    publisher = {Zenodo},
    doi = {10.5281/zenodo.19391349},
}

@misc{supplement,
    author = {Zielinski, Laura and Hsu, Justin},
    title = {Supplemental Material for Article `Synthesizing Backward Error Bounds, Backward'},
    year = {2026},
    publisher = {Association for Computing Machinery},
    doi = {10.1145/3808333},
}
